\documentclass[10pt,sigconf, nonacm=true]{acmart}
\pdfoutput=1

\usepackage{booktabs} 
\usepackage{multirow}
\usepackage{subfigure}
\usepackage{dsfont} 
\usepackage[linesnumbered,lined,boxed,commentsnumbered, ruled]{algorithm2e}
\usepackage{makecell}
\usepackage[utf8]{inputenc}
\usepackage[T1]{fontenc}
\usepackage{tabularx}
\usepackage{tikz}
\usepackage{tablefootnote}
\usepackage{xcolor}
\usepackage{amsthm}
\usepackage{amsmath}
\usepackage{amsfonts}
\usepackage{marginnote}
\usepackage{amssymb} 
\usepackage{diagbox}

\AtBeginDocument{%
  \providecommand\BibTeX{{%
    \normalfont B\kern-0.5em{\scshape i\kern-0.25em b}\kern-0.8em\TeX}}}

\iftrue 
\fi

\newtheorem{dft}{Definition}

\newtheorem{thm}{Theorem}
\newtheorem{lmm}{Lemma}

\reversemarginpar
\renewcommand\footnotetextcopyrightpermission[1]{} 
\begin{document}
\title{Regular Path Query Evaluation on Streaming Graphs}
\titlenote{A shorter version of this paper has been accepted for publication in \emph{2020 International Conference
on Management of Data (SIGMOD ’20)}.}

\iftrue 
\author{Anil Pacaci}

\affiliation{%
  \institution{University of Waterloo}
}

\email{apacaci@uwaterloo.ca}

\author{Angela Bonifati}

\affiliation{%
  \institution{Lyon 1 University}
}

\email{angela.bonifati@univ-lyon1.fr}

\author{M. Tamer \"{O}zs{u}}

\affiliation{%
  \institution{University of Waterloo}
}
\email{tamer.ozsu@uwaterloo.ca}

\fi 


\begin{abstract}
    
    
    We study persistent query evaluation over streaming graphs, which is becoming increasingly important. We focus on navigational queries that determine if there exists  a path  between two entities that satisfies a user-specified constraint.
    We adopt the Regular Path Query (RPQ) model that specifies navigational patterns with labeled constraints. We propose deterministic algorithms to efficiently evaluate persistent RPQs under both \textit{arbitrary} and \textit{simple} path semantics in a uniform manner.  Experimental analysis on real and synthetic streaming graphs shows that the proposed algorithms can process up to tens of thousands of edges per second and efficiently answer RPQs that are commonly used in real-world workloads. 
\end{abstract}

\maketitle


\section{Introduction}
\label{sec:intro}

Graphs are used to model complex interactions in various domains ranging from social network analysis to communication network monitoring, from retailer customer analysis to bioinformatics.
%
Many real-world applications generate graphs over time as new edges are produced resulting in \textit{streaming graph}s  \cite{sahu:2018aa}. 
Consider an e-commerce application: each user and item can be modelled as a vertex and each user interaction such as clicks, reviews, purchases can be modelled as an edge.
The system receives and processes a sequence of graph edges (as users purchase items, like them, etc). 
These graphs are unbounded, and the edge arrival rates can be very high:  Twitter’s recommendation system ingests 12K events/sec on average \cite{grewal2018recservice}, Alibaba's user-product graph processes 30K edges/sec at its peak \cite{qiu2018real}. Recent experiments show that existing graph DBMSs are not able to keep up with the arrival rates of many real streaming graphs  \cite{pacaci:2017aa}.

Efficient querying of streaming graphs is a crucial task for applications that monitor complex patterns and, in particular, \textit{persistent queries} that are registered to the system and whose results are generated incrementally as the graph edges arrive.
Querying streaming data in real-time imposes novel requirements in addition to challenges of graph processing: (i) graph edges arrive at a very high rate and real-time answers are required as the graph emerges, and (ii) graph streams are unbounded, making it infeasible to employ batch algorithms on the entire stream. 
Most existing work focus on the \textit{snapshot} model, which assumes that graphs are static and fully available, and adhoc queries reflect the current state of the database (e.g., \cite{cohen2003, yakovets2016query, vldb10_yildirim:2010, icde13_6544893, su2016reachability, Wadhwa:2019:EAR:3299869.3319882, koschmieder2012regular}).
The \textit{dynamic graph} model addresses the evolving nature of these graphs; however, algorithms in this model assume that the entire graph is fully available and they compute how the output changes as the graph is updated \cite{kapron2013dynamic, bernstein2016maintaining, roditty2016fully, lkacki2011improved}.

In this paper, we study the problem of persistent 
query processing over streaming graphs, addressing the limitations of existing approaches.
We adopt the Regular Path Query (RPQ) model that focuses on  \textit{path navigation}, e.g., finding pairs of users in a network connected by a path whose label (i.e., the labels of edges in the path) matches path constraints. 
RPQ specifies path constraints that are expressed using a regular expression over the alphabet of edge labels and checks whether a path exists with a label that satisfies the given regular expression  \cite{mendelzon1995finding, baeza13}. The RPQ model provides the basic navigational mechanism to encode graph queries, striking a balance between expressiveness and computational complexity \cite{angles2017foundations, bonifati2018querying,  sparql11-w3c, angles2018g, grades2016_van2016pgql}.
Consider the streaming graph of a social network application presented in Figure \ref{fig:stream_graph}. 
The query $Q_1: (\textit{follows} \circ \textit{mentions})^+$ in Figure \ref{fig:query_graph} represents a pattern for a real-time notification query where  user \texttt{x} is notified of other users who are connected by a path whose edge labels are even lengths of alternating $\textit{follows}$ and $\textit{mentions}$.
At time $t=18$, the pair of users $(x,y)$  is connected by such a path, shown by bold edges 
in Figure \ref{fig:data_graph}.

\begin{figure*}
    \centering
    \subfigure[Streaming Graph $S$]{
        \includegraphics[width=2.2in]{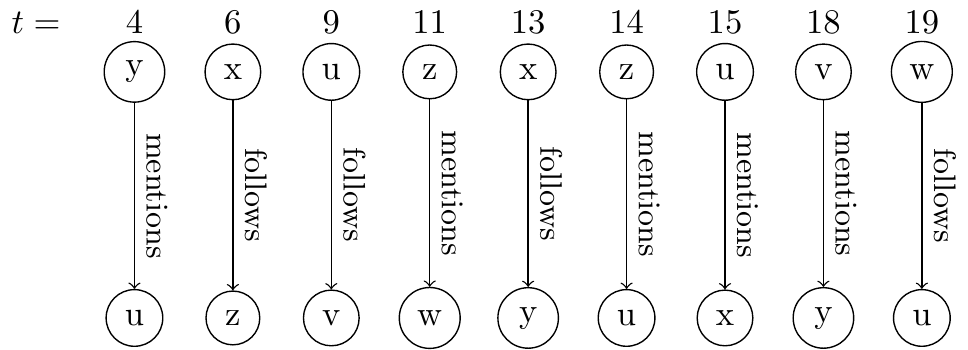}
        \label{fig:stream_graph}
    }
        \subfigure[Snapshot Graph $G$]{
        \includegraphics[width=1.4in]{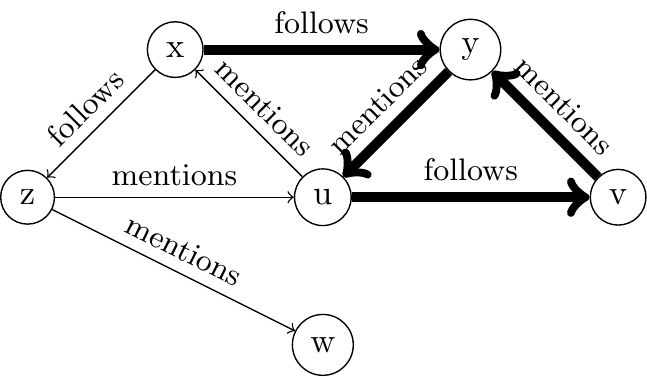}
        \label{fig:data_graph}
    }
    \subfigure[Query Graph $Q_1$]{
        \includegraphics[width=1.4in]{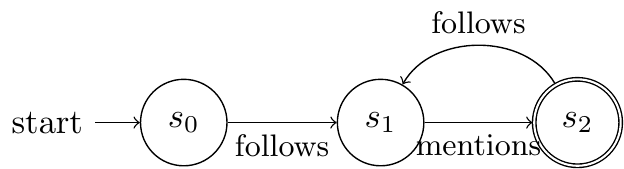}
        \label{fig:query_graph}
    }
    \subfigure[Product Graph $P_{G,A}$]{
        \includegraphics[width=1.5in]{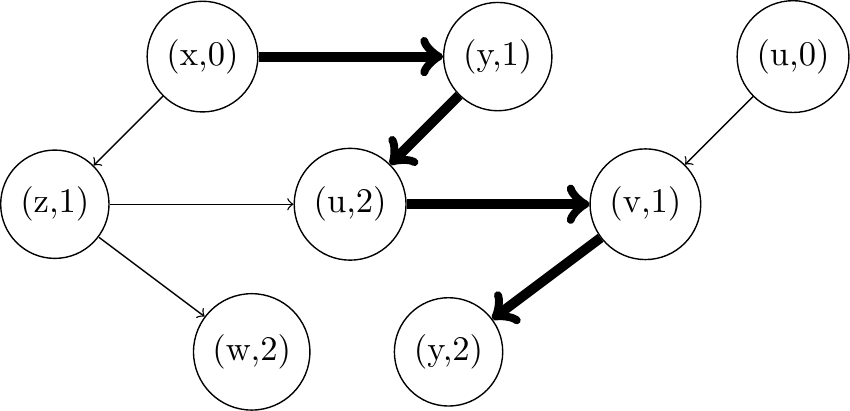}
        \label{fig:product_graph}
    }
    \caption{(a) A streaming graph $S$ of a social networking application, (b) the snapshot of $S$ at $t=18$, (c) automaton for the query $Q_1: (\textit{follows} \circ \textit{mentions})^+$, and (d) the product graph $P_{G,A}$ .}
    \label{fig:running_example}
\end{figure*}

It is known that for many streaming algorithms the space requirement is lower bounded by the stream size \cite{bbd+02}. Since the stream is unbounded, deterministic RPQ evaluation is infeasible without storing all the edges of the graph (by reduction to the length-2 path problem that is infeasible in sublinear space \cite{feigenbaum2005graph}).
In streaming systems, a general solution for bounding the space requirement is to evaluate queries on a \textit{window} of  data from the stream. In a large number of applications, focusing on the most recent data is  desirable. 
Thus, the  windowed evaluation model not only provides a tool to process unbounded streams with bounded memory but also restricts the scope of queries on recent data, a desired feature in many streaming applications. 
In this paper we consider the \textit{time-based sliding window} model where a fixed size (in terms of time units) window is defined that slides at well-defined intervals \cite{golabo03a}. In our context,  new graph edges enter the window during the window interval, and when the window slides, some of the ``old'' edges leave the window (i.e., expire). Managing this window processing as part of RPQ evaluation is challenging and our solutions address the issue in a uniform manner.

In this paper, for the first time, we study the design space of persistent RPQ evaluation algorithms in two main dimensions: the  path semantics they support and the result semantics based on application requirements.
Along the first dimension, we propose efficient incremental algorithms for both \textit{arbitrary} and \textit{simple} path semantics.
The former 
allows a path to traverse the same vertex multiple times, whereas under the latter semantics a path cannot traverse the same vertex more than once  \cite{angles2017foundations}.
Consider the example graph given in Figure \ref{fig:data_graph}; the sequence of vertices $\langle x,y,u,v,y \rangle$ is a valid path for query $Q_1$ with arbitrary path semantics whereas the simple path semantics does not traverse this path as it visits vertex $y$ twice. 
Along the second dimension, we consider \textit{append-only} streams where tuples in the window expire only due to window movements, then extend our algorithms to support \textit{explicit deletions} to deal with cases where users/applications might explicitly delete a previously arrived edge. We use the \textit{negative tuples} approach \cite{golab:2010fj}  to process explicit deletions.
Table {\ref{tab:amortized}} presents the combined complexities of the proposed algorithms in each quadrant in terms of amortized cost.
\begin{table}
\caption{Amortized time complexities of the proposed algorithms for a streaming graph $S$ with $m$ edges and $n$ vertices and RPQ $Q_R$ whose  automata has $k$ states.}
\small
	\centering
	\begin{tabular}{ | r | c | c | }
	\hline
       \diagbox{\footnotesize Path\\Semantics}{\footnotesize Result\\Semantics} & Append-Only &  Explicit Deletions \\
	\hline
	Arbitrary (\S \ref{sec:arbitrary}) & $\mathcal{O}(n \cdot k^2)$ & $\mathcal{O}(n^2 \cdot k)$  \\
	Simple\tablefootnote{These results hold in the absence of conflicts, a condition on cyclic structure of the query and graph that is precisely defined in \S \ref{sec:simple-insertonly}.}  (\S \ref{sec:simple}) & $\mathcal{O}(n \cdot k^2)$ & $\mathcal{O}(n^2 \cdot k)$. \\
	\hline
	\end{tabular}
\label{tab:amortized}
\end{table}

To the best of our knowledge, these are the first streaming algorithms to address RPQ evaluation on sliding windows over streaming graphs under both arbitrary (\S \ref{sec:arbitrary}) and simple path semantics (\S \ref{sec:simple}).
Our proposed algorithm for streaming RPQ evaluation under arbitrary path semantics incrementally maintains results for a query $Q_R$ on a sliding window $W$ over a streaming graph $S$ as new edges enter and old edges expire due to window slide.
We follow the implicit window semantics, where  newly arriving edges are processed as they arrive (and new results appended to the output stream) while the removal of  expired edges occur at user-specified slide intervals.
We then turn our attention to simple path semantics (\S \ref{sec:simple}). 
The static version of the RPQ evaluation problem is NP-hard in its most general form \cite{mendelzon1995finding}, which has caused existing work to focus only on arbitrary path semantics.
Yet, it is proven to be tractable when restricted to certain classes of regular expressions or by imposing restrictions on the graph instances \cite{mendelzon1995finding, baganbg13}.
A recent analysis \cite{bonifati2017analytical, bonifati2019navigating} of real-world SPARQL logs shows that a large portion of RPQs posed by users does indeed fall into those tractable classes, motivating the design of efficient algorithms for streaming RPQ evaluation under simple path semantics.
Our proposed algorithm admits efficient solutions for streaming RPQs under simple path semantics in the absence of conflicts, a condition on the cyclic structure of graphs that enables efficient batch algorithms (precisely defined in \S \ref{sec:simple-insertonly}) \cite{mendelzon1995finding}.
Indeed, this algorithm has the same  amortized time complexity  as the proposed algorithm for arbitrary path semantics under the same condition.
The proposed algorithms 
incrementally maintain query answers as the window slides thus eliminating the computational overhead of the naive strategy of batch computation after each window movement.
Furthermore, they support negative tuples to accommodate applications where users might explicitly delete a previously inserted edge.
Albeit relatively rare, explicit deletions are a desired feature of real-world applications that process and query streaming graphs, and it is known to require special attention \cite{golabo05}.
We show that window management and explicit deletions can be handled 
in a uniform manner using the same machinery 
(\S \ref{sec:arbitrary-delete}). 
Finally, we empirically evaluate the performance of our proposed algorithms using a variety of real-world and synthetic streaming graphs on real-world RPQs that cover more than 99\% of all recursive queries abundantly found in massive Wikidata query logs \cite{bonifati2019navigating} (\S \ref{sec:experimens}).



\section{Preliminaries}
\label{sec:preliminaries}

\label{sec:definitions}

\begin{dft}[Graph]\label{def:graph}
A directed labeled graph is a quin\-tuple $G = (V, E, \Sigma, \psi, \phi)$ where $V$ is a set of vertices, 
$E$ is a set of edges, 
$\Sigma$ is a set of labels,
$\psi: E \rightarrow V \times V$ is an incidence function
and $\phi: E \rightarrow \Sigma$ is an edge labelling function. 
\end{dft}

\begin{dft}[Streaming Graph Tuple]\label{def:insertstreamtuple}
A streaming graph tuple (sgt) $t$ is a quadruple $(\tau, e,l, op)$ where $\tau$ is the event (application) timestamp of the tuple assigned by the data source, $e=(u,v)$ is the directed edge with source vertex $u$ and target vertex $v$, $l \in \Sigma$ is the label of the edge $e$ and $op$ is the type of the edge, i.e., insert ($+$) or delete ($-$) .
\end{dft}

\begin{dft}[Streaming Graph]\label{def:insertstream}
 
\sloppy A streaming graph $S$ is a constantly growing sequence of streaming graph tuples (sgts) $S= \langle t_1, t_2, \cdots, t_m \rangle$ in which each tuple $t_i$ arrives at a particular time $\tau_i$ ($\tau_i < \tau_j$ for $i<j$). 
\end{dft}

In this paper, we assume that sgts\footnote{We use ``sgt" and ``tuple'' interchangeably.
} are generated by a single source and arrive in source timestamp order $\tau_i$, which
 defines their ordering in the stream. We leave the problem of 
out-of-order delivery as future work.





\begin{dft}[Time-based Window]\label{def:win}

A time-based window $W$  over a streaming graph $S$ is defined by a time interval $(W^b,W^e]$ where $W^b$ and $W^e$ are the beginning and end times of window $W$ and $W^e - W^b = |W|$. The window contents $W(c)$ is the multiset of sgts where the timestamp $\tau_i$ of each sgt $t_i$ is in the window interval, i.e.,  $W(c) = \{ t_i \mid W_b < \tau_i \leq W_e  \}$.\footnote{We use $W$ interchangeably to refer to a window interval or its contents.}
\end{dft}

\begin{dft}[Time-based Sliding Window]\label{def:window}
A time-based sliding window $W$ with a slide interval $\beta$ is a time-based window that progresses every $\beta$ time units.
At any time point $\tau$, a time-based sliding window $W$ with a slide interval $\beta$ defines a time interval $(W^b,W^e]$ where $W^e = \lfloor \tau / \beta \rfloor \cdot \beta$ and $W^b = W^e - |W|$.
The contents of $W$ at time $\tau$ defines a snapshot graph $G_{W, \tau} = (V_{W, \tau}, E_{W, \tau}, \Sigma_{W, \tau}, \psi, \phi)$ where $E_{W, \tau}$ is the set of all edges that appear in sgts in $W$ and $V_{W, \tau}$ is the set of vertices that are endpoints of edges in $E_{W, \tau}$.
\end{dft}

Figure {\ref{fig:stream_graph}} shows an excerpt of a streaming graph $S$ at $t = 19$.
Figure {\ref{fig:data_graph}} shows the snapshot graph $G_{W,18}$ defined by window $W$ with $|W|=15$ over this graph $S$.


A time-based sliding window $W$ might progress either at every time unit, i.e. $\beta=1$ (\textit{eager evaluation}; resp. \textit{expiration}) or at $\beta>1$ intervals (\textit{lazy evaluation}; resp. \textit{expiration}) \cite{patroumpas2006window}.
Eager evaluation 
produces fresh results but windows can be expired lazily if queries do not produce premature expirations \cite{golabo05}.
We use eager evaluation ($\beta = 1$) but lazy expiration ($\beta > 1$) as it enables us to separate window maintenance from processing of incoming sgts (\S \ref{sec:arbitrary-insertonly}).  

\begin{dft}[Path and Path Label]\label{def:path}
\sloppy Given $u, v \in V$, a path $p$ from $u$ to $v$ in graph $G$ is a sequence of edges $ u \stackrel{p}{\rightarrow} v : \langle (v_0, l_0,v_1), \cdots, (v_{n-1},l_{n-1}, v_n) \rangle$ where $v_0 = u$ and $v_n = v$.
The label of a path $p$ is denoted by $\phi(p) = l_0l_1\cdots l_{n-1} \in \Sigma^*$.
\end{dft}

\begin{dft}[Regular Expression \& Regular Language]\label{def:reglanguage}

A regular expression $R$ over an alphabet $\Sigma$ is defined as $ R ::= \epsilon \mid a \mid \ R \circ R \mid R + R \mid R^*$ where (i) $\epsilon$ denotes the empty string, (ii) $a \in \Sigma$ denotes a character in the alphabet, (iii) $\circ$ denotes the concatenation operator, (iv) $+$ denotes the alternation operator, and (v) $*$ represents the Kleene star. 
We use $\neg$ to denote the negation of an expression, and  
$R^+$ to denote 1 or more repetitions of $R$.
A regular language $L(R)$ is the set of all strings that can be described by the regular expression $R$.
\end{dft}

\begin{dft}[Regular Path Query -- RPQ]
A Regular Path Query $Q_R$ asks for pairs of vertices $(u,v)$ that are connected by a path $p$ from $u$ to $v$ in graph $G$, where the path label $\phi(p)$ is a word in the regular language defined by the regular expression $R$ over the graph's edge labels $\Sigma$, i.e., $\phi(p) \in L(R)$. Answer to  query $Q_R$ over  $G$, $Q_R(G)$, is the set of all pairs of vertices that are connected by such paths.
\end{dft}



Sliding windows adhere to two alternative semantics: \textit{implicit} and \textit{explicit} {\cite{golab:2010fj}}.
Implicit windows add new results to query output as new sgts arrive and do not invalidate the previously reported results upon their expiry as the window moves.
In the absence of explicit edge deletions, the query results are monotonic.
Under this model, the result set of a streaming RPQ over a streaming graph $S$ and a sliding window $W$ at time $\tau$ contains all paths in all previous snapshot graphs $G_{W,\pi}$ where $0 < \pi \leq \tau$, i.e., $Q_R(S, W, \tau) = \bigcup_{0 < \pi \leq \tau} Q_R(G_{W,\pi})$.
Alternatively, explicit windows remove previously reported results involving tuples (i.e., sgts) that have expired from the window; hence, persistent queries with explicit windows are akin to incremental view maintenance.
Under this model, the result set of a streaming RPQ over a streaming graph $S$ and a sliding window $W$ at time $\tau$ contains only the paths in the snapshot $G_{W,\tau}$ of the streaming graph, i.e., $Q_R(S, W, \tau) = Q_R(G_{W,\tau})$.
Explicit windows, by definition, produce non-monotonic results as previous results are negated when the window moves \cite{golab:2010fj}.
We employ the implicit window model in this paper as it enables us to preserve the monotonicity of query results and produce an append-only stream of query results (in the absence of explicit deletions).

\begin{dft}[Streaming RPQ]
\label{def:streamingRPQ}
A streaming RPQ is 
defined over a streaming graph $S$ and a sliding window $W$.
A pair of vertices $(u,v)$ is an answer for a streaming RPQ, $Q_R$,  at time $\tau$ if there exists a path $p$ between $u$ and $v$ in $G_{W,\tau}$, i.e.,  all edges in $p$ are in window $W$.
We define the timestamp $p.ts$ of a path $p$ as the minimum timestamp among all edges of $p$.
Under the implicit window model, the result set of a streaming RPQ $Q_R$ over a streaming graph $S$ and a sliding window $W$ is an append-only stream of pairs of vertices
$(u,v)$ where there exists a path $p$ between $u$ and $v$ with label $\phi(p) \in L(R)$ and all the edges in $p$ are at most one window length, i.e., $|W|$ time units, apart. Formally:
\begin{align*}
Q_R(S, W, \tau) = \{ (u,v) \mid &  \exists p: u \stackrel{p}{\rightarrow} v  \land  \phi(p) \in L(R) \land \\ &  \max_{e \in p}{(e.ts)} < p.ts + |W| \leq \tau  \}  
\end{align*}

\end{dft}


\begin{dft}[Deterministic Finite Automaton]
Given a regular expression $R$, $A = (S, \Sigma, \delta, s_0, F)$ is a Deterministic Finite Automaton (DFA) for $L(R)$ where $S$ is the set of states, 
$\Sigma$ is the input alphabet, 
$\delta \colon S \times \Sigma \rightarrow S$
is the state transition function, $s_0 \in S$ is the start state and $F \subseteq S$ is the set of final states. 
$\delta^*$ is the extended transition function defined as:
\[ \delta^*(s, w \circ a) = {\delta(\delta^*(s,w),a)} \]
\noindent where $s \in S$, $a \in \Sigma$, $w \in \Sigma^*$, and $\delta^*(s,\epsilon)=s$ for the empty string $\epsilon$. 
We say that a word $w$ is in the language accepted by $A$ if $\delta^*(w,s_0) = s_f $ for some $ s_f \in F$.
\end{dft}

\begin{dft}[Product Graph]
\label{def:product_graph}
Given a graph $G = (V,E,\Sigma,\phi)$ and a DFA $A = (S, \Sigma, \delta, s_0, F)$, we define the product graph $P_{G,A} = (V_P, E_P, \Sigma, \phi_{P})$ where $V_P = V \times S$, $E_P \subseteq V_P \times V_P$, and $((u,s),(v,t))$ is in $E_P$ iff $(u,v) \in E$ and $\delta(s, \phi(u,v))=t$. 
\end{dft}

Figure {\ref{fig:product_graph}} shows the product graph of  $G_{W,18}$ (Figure {\ref{fig:data_graph}}) and the DFA $A$ of the query $Q_1$ (Figure {\ref{fig:query_graph}}).


For a given RPQ, $Q_R$, we first use Thompson's construction algorithm~{\cite{thompson1968programming}} to create a NDFA that recognizes the language $L(R)$, then create the equivalent minimal DFA, $A$, using Hopcroft's algorithm~{\cite{hopcroft1971n}}.
In the rest of the paper, we use $A$ and the product graph $P_{G,A}$ to describe the proposed algorithms for RPQ evaluation in the streaming graph model. 



\section{RPQ with Arbitrary Semantics}
\label{sec:arbitrary}

In this section, we study the problem of  RPQ evaluation over sliding windows of streaming graphs under arbitrary path semantics, that is, finding pairs of vertices $u,v \in V$ where (i) there exists a (not necessarily simple) path $p$ between $u$ and $v$ with a label $\phi(p)$ in the language $L(R)$, and (ii) timestamps of all edges in  path $p$ are in the range of  window $W$.
We first consider append-only streams where the query results are monotonic (under \textit{implicit window} model) such that existing results do not expire from the result set when input tuples expire from the window \cite{golab:2010fj}.
Then, we show how the proposed algorithms are extended to support negative tuples to handle explicit edge deletions.

\textbf{Batch Algorithm}:
RPQs can be evaluated in polynomial time under arbitrary path semantics \cite{mendelzon1995finding}. 
Given a product graph $P_{G,A}$, there is a path $p$ in $G$ from $x$ to $y$ with label $w$ that is in $L(R)$ if and only if there is a path in $P_{G,A}$ from $(x,s_0)$ to $(y,s_f)$, where $s_f \in F$. 
The batch RPQ evaluation algorithm  under arbitrary path semantics traverses the product graph $P_{G,A}$ by simultaneously traversing graph $G$ and the automaton $A$. 
The time complexity of the batch algorithm is {$\mathcal{O}$}$(n\cdot m\cdot k^2)$ under the assumption that there are more edges than isolated vertices in  $G$.

\subsection{RPQ over Append-Only Streams}
\label{sec:arbitrary-insertonly}
We first present an incremental algorithm for Regular Arbitrary Path Query (RAPQ) evaluation over append-only streams. As noted above, using implicit window semantics, RAPQs are monotonic,
i.e., $Q_R(S, W, \tau) \subseteq Q_R(S, W,\tau+\epsilon)$ for all $\tau, \epsilon \geq 0$. 
Algorithm \textbf{\ref{alg:arbitrary}} consumes a sequence of append-only tuples (i.e., \texttt{op} is $+$), and simultaneously traverses the product graph of the snapshot graph $G_{W,\tau}$ of the window $W$ over a graph stream $S$ and the automaton $A_{\tau}$ of $Q_R$ for each tuple $t_{\tau}$, and it produces an append-only stream of results for $Q_R(S, W, \tau)$. 
As in the case of the batch algorithm, such traversal of $G_{W,\tau}$ guided with the automaton $A$ emulates a traversal of the product graph $P_{G,A}$.

\begin{algorithm}
\small
\SetAlgoRefName{RAPQ}
	\SetKwData{Left}{left}\SetKwData{This}{this}\SetKwData{Up}{up}
	\SetKwInOut{Input}{input}\SetKwInOut{Output}{output}
	\Input{ 
	        Incoming tuple  $t_{\tau} = (\tau,e_{\tau},l,op), e_{\tau}=(u,v)$, \\
	}
	$G_{W,\tau} \leftarrow G_{W,\tau-1}$ (op) $e_{\tau}$ \label{line:insertrapq_addwindow} \\
    \ref{alg:insertrapq_expiry}($G_{W, \tau}, T_x, \tau$) $\forall T_x \in \Delta$ \label{line:insertrapq_callexpiry} \tcp{on user-defined slide intervals}
    set of results $R \leftarrow \emptyset $ 
    
    
    \ForEach{$T_x \in \Delta$}{ \label{line:insertrapq_foreachtree}
        \ForEach{$s,t \in S$ where $t=\delta(s,l)$}{ \label{line:insertrapq_foreachtransition}
            \If{$(u,s) \in T_x \wedge (u,s).ts > \tau - |W| $}{ \label{line:insertrapq_windowedge}
                \If{$(v,t) \not\in T_x \vee (v,t).ts < min((u,s).ts, \tau)$}{ \label{line:insertrapq_updateexpiry}
                    $R \leftarrow R +$\ref{alg:insertrapq_insert_edge}($T_x$, $(u,s)$, $(v,t)$, $e=(u,v)$) \label{line:insertrapq_callinsert}
                }
            }
        }
    }
	
	$Q_R(S, W, \tau) \leftarrow Q_R(S, W, \tau-1) + R$
	\caption{}\label{alg:arbitrary}
\end{algorithm} 

\begin{algorithm}
\small
\SetAlgoRefName{Insert}
	\SetKwData{Left}{left}\SetKwData{This}{this}\SetKwData{Up}{up}
	\SetKwInOut{Input}{input}\SetKwInOut{Output}{output}
	\Input{Spanning Tree $T_x$ rooted at $(x,s_0)$, \\
	        parent node $(u,s)$, child node $(v,t)$, Edge $e=(u,v)$, \\
	}
	\Output{
	The set of results $R$
	}
	$R \leftarrow \emptyset$ \\
	$(v,t).pt = (u,s)$ \\ \label{line:insertrapq_updateparent}
	$(v,t).ts = min(e.ts, (u,s).ts)$ \label{line:insertrapq_updatetimestamp}

	\If{$(v,t) \not\in T_x$ \label{line:insertrapq_invariant2} }{ 
    	\If{$t \in F$}{
            $R \leftarrow R + (x,v)$ \label{line:addresult}
        }
        \ForEach{edge $(v,w) \in W_{G, \tau} $ s.t. $\delta(t, \phi(v,w))=q$}{
            \If{$(w,q) \not\in T_x \vee (w,q).ts < min((v,t).ts, (v,w).ts)$}{ \label{line:insertrapq_exploreedge}
                $R \leftarrow R +$\ref{alg:insertrapq_insert_edge}($T_x$, $(v,t)$, $(w,q)$, $e=(v,w)$) \label{line:insertrapq_callinsert2}
            }
    	}
	}

	\Return $R$

\caption{}\label{alg:insertrapq_insert_edge}
\end{algorithm}

\begin{dft}[$\Delta$ Tree Index]\label{def:delta}
Given an automaton $A$ for a query $Q_R$ and a snapshot $G_{W,\tau}$ of a streaming graph $S$ at time $\tau$, $\Delta$ is a collection of spanning trees where each tree $T_x$ is rooted at a vertex  $x \in G_{W,\tau}$ for which there is a corresponding node in the product graph of $A$ and $G_{W,\tau}$ with the start state $s_0$, i.e., $\Delta = \{ T_x \mid x \in G_{W,\tau} \wedge (x,s_0) \in V_{P_{G,A}} \}$. 

\end{dft}

In the remainder, we use the term ``vertex'' to denote endpoints of sgts, and the term ``node'' to denote vertex-state pairs in spanning trees.

A node $(u,s) \in T_x$ at time $\tau$ indicates that there is a path $p$ in $G_{W,\tau}$ from $x$ to $u$ with label $\phi(p)$ and timestamp $p.ts$ such that $\delta^*(s_0,\phi(p))=s$ and $(\tau - |W|) < p.ts \leq \tau$, i.e., word $\phi(p) \in \Sigma^*$ takes the automaton $A_{\tau}$ from the initial state $s_0$ to a state $s$ and the timestamp of the path is in the window range.
Each node $(u,s)$ in a tree $T_x$ maintains a pointer $(u,s).pt$ to its parent in $T_x$.
Additionally, the timestamp $(u,s).ts$ is the minimum timestamp among all edges in the path from $(x,s_0)$ to $(u,s)$ in the spanning tree $T_x$, following Definition \ref{def:streamingRPQ}.

The proposed algorithm  continuously updates  $G_{W,\tau}$  upon arrival of new edges and expiry of old edges.
In addition to $G_{W,\tau}$, it maintains a tree index ($\Delta$) to support efficient incremental RPQ evaluation that enables efficient RPQ evaluation on sliding windows over streaming graphs. 

\begin{figure}
    \centering
    \subfigure[t=18]{
        \includegraphics[height=2in]{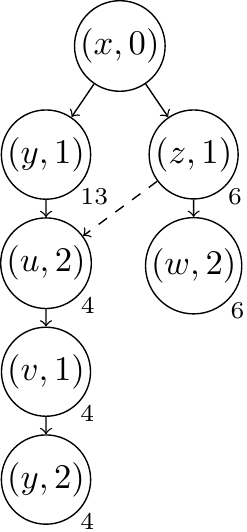}
        \label{fig:spanning_tree_before_insert}
    }
    \qquad
    \subfigure[t=19]{
        \includegraphics[height=2in]{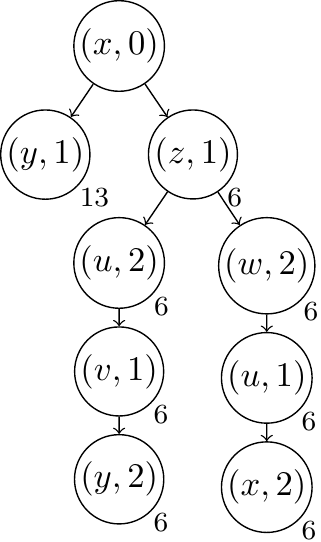}
        \label{fig:spanning_tree_after_insert}
    }
    \caption{A spanning tree $T_x \in \Delta$ for the example given in Figure \ref{fig:running_example} rooted at $(x,0)$ (a) before and (b) after the edge $e=(w,u)$ with label $follows$ at $t=19$ is consumed. The timestamp of each node given at the corner.}
    \label{fig:spanning_tree}
\end{figure}

\begin{example}
Figure \ref{fig:spanning_tree_before_insert} illustrates a spanning tree $T_x \in \Delta$ for the streaming graph $S$ and the RPQ $Q_1$ given in Figure \ref{fig:running_example} at time $t=18$.
The tree
in Figure \ref{fig:spanning_tree_before_insert} is constructed through a traversal of the product graph starting from  node $(x,0)$, visiting nodes $(y,1)$, $(u,2)$, $(v,1)$ and $(y,2)$, forming the path from the root to the node $(y,2)$ in Figure \ref{fig:spanning_tree_before_insert}.
Similar to the batch algorithm, this corresponds to the traversal of the path $\langle x,y,u,v,y \rangle$ in the snapshot of the streaming graph (Figure \ref{fig:data_graph}) with label $\langle \textit{follows, mentions$, $follows, mentions} \rangle$ taking the automaton from state $0$ to $2$ through the path $\langle 0,1,2,1,2 \rangle$  in the corresponding automaton (Figure \ref{fig:query_graph}).
The timestamp of the node $(y,2) \in T_x$ at $t = 18$ is $4$ as the edge with the minimum timestamp on the path from the root is $(y,\textit{mentions},u)$ with $\tau=4$.
\end{example}

\begin{lmm}\label{lem:arbitraryinvariants}
The proposed Algorithm \textbf{\ref{alg:arbitrary}} maintains the following two  invariants of the  $\Delta$ tree index: 

\begin{enumerate}
    \item A node $(u,s)$ with timestamp $ts$ is in $T_x$ if there exists a path $p$ in $G_{W,\tau}$ from $x$ to $u$ with label $\phi(p)$ and timestamp $(u,s).ts$ such that $s=\delta^*(s_0, \phi(p))$ and $(u,s).ts = p.ts \in (\tau - |W|, \tau ]$, i.e., there exists a path $p$ in $G_{\tau}$ from $x$ to $u$ with label $\phi(p)$ such that $\phi(p)$ is a prefix of a word in $L(R)$ and all edges are in the window $W$.
    \item At any given time $\tau$, a node $(u,s)$ appears in a spanning tree $T_x$ at most once with a timestamp in the range $(\tau - W, \tau ]$.
\end{enumerate}
\end{lmm}

\begin{proof}
First, we show that Algorithm \textbf{\ref{alg:insertrapq_expiry}} maintains the two invariants of the $\Delta$ tree index.
The second invariant is preserved as Algorithm \textbf{\ref{alg:insertrapq_expiry}} does not add any node to a spanning tree $T_x \in \Delta$.
For each spanning tree $T_x \in \Delta$, Line \ref{line:insertrapq_candidatenodes} of the algorithm identifies the set of nodes that are potentially expired at time $\tau$, $P = \{ (v,t) \in T_x \mid (v,t).ts \leq \tau - |W| \}$. 
Initially, all expired nodes are removed from the spanning tree $T_x$ (Line \ref{line:insertrapq_candidatesubgraph}).
Algorithm \textbf{\ref{alg:insertrapq_insert_edge}} is invoked for each expired node $(v,t) \in P$ if there exists a valid edge in the window $G_{W, \tau}$ from another valid node in $T_X$ (Line \ref{line:insertrapq_connected}).
Finally, nodes that are reconnected to the spanning tree $T_x$ by Algorithm \textbf{\ref{alg:insertrapq_insert_edge}} are removed from $P$ as there exists an alternative path from the root through $(u,s)$.
As a result, Algorithm \textbf{\ref{alg:insertrapq_expiry}} removes a node $(v,t)$ from the spanning tree $T_x$ if there does not exist any path $p$  in $G_{W,\tau}$ from $x$ to $u$ with a label $l$ such that $s=\delta^*(s_0,l)$ and $p.ts > \tau - |W|$, preserving the first invariant.

It is easy to see that the second invariant is preserved after each call to Algorithm \textbf{\ref{alg:arbitrary}} given that Algorithm \textbf{\ref{alg:insertrapq_expiry}} preserves both invariants. 
The second invariant is preserved as 
Line \ref{line:insertrapq_invariant2} of Algorithm \textbf{\ref{alg:insertrapq_insert_edge}} adds the node $(v,t)$ to a spanning tree $T_x$ only if it has not been previously inserted. 

We show that Algorithm \textbf{\ref{alg:arbitrary}} preserves the first invariant by induction on the length of the path. 
For the base case $n = 1$, consider that $t_{\tau} = (\tau, e, l, +), e=(u,v)$ arrives in the window $W$ at time $\tau$.
Line \ref{line:insertrapq_foreachtransition} in  Algorithm \textbf{\ref{alg:arbitrary}} identifies each state $t$ where there is a transition from the initial state $s_0$ with label $l$, i.e., $\delta(s_0,l)=t$. 
The path from $(u,s_0)$ to $(v,t)$ is added to $T_x$ with $(v,t).ts = \tau$.
For the non-base case, consider a node $v \in G_{W,\tau}$ where there exists a path $p$ of length $n$ from $x$ where $t=\delta^*(s_0,\phi(p))$ and $p.ts > \tau - |W|$. 
Let $(u,s)$ be the predecessor of $(v,t)$ in the path, that is edge $(u,v)$ is in $G_{W,\tau}$ with label $l$ and $\delta(s,l)=t$. 
By the inductive hypothesis, the node $(u,s)$ is in  $T_x$  as there exists a path $q$ of length $n-1$ from $x$ to $u$ in $G_{W,\tau}$ where $s=\delta^*(s_0, \phi(q))$ and $q.ts > \tau - |W|$. 
If the edge $e=(u,v) \in G_{W,\tau}$ is already in the window $W$ ($\tau - |W| < e.ts < \tau$) when the node $(u,s)$ is inserted into $T_x$ , then  the proposed algorithm invokes Algorithm \ref{alg:insertrapq_insert_edge} with node $(u,s)$ as parent and node $(v,t)$ as child (Line \ref{line:insertrapq_callinsert}) and its adds $(v,t)$ into $T_x$ with timestamp $(v,t).ts = min(e.ts, (u,s).ts)$ (Line \ref{line:insertrapq_updatetimestamp}). 
If the edge $e=(u,v)$ is processed by the proposed algorithm after the node $(u,s)$ is inserted in $T_x$ ($e.ts > (u,s).ts$), then Line \ref{line:insertrapq_callinsert2} in Algorithm \textbf{\ref{alg:insertrapq_insert_edge}} guarantees that Algorithm \textbf{\ref{alg:insertrapq_insert_edge}} is invoked with the node $(v,t)$. 
Lines \ref{line:insertrapq_updateparent} and \ref{line:insertrapq_updatetimestamp} in Algorithm \textbf{\ref{alg:insertrapq_insert_edge}} adds the node $(v,t)$ to $T_x$, and properly updates its parent pointer to $(u,s)$ and its timestamp $(v,t).ts = min(e.ts, (u,s).ts)$.
The first invariant is preserved in either case as  $\tau - |W| < p.ts = (v,t).ts \leq \tau$.
%
Therefore we conclude that Algorithm \textbf{\ref{alg:arbitrary}} also preserves the first invariant.
\end{proof}

The first invariant allows us to trace all reachable nodes from a root node $(x,s_0)$ whereas 
the second invariant prevents Algorithm \textbf{\ref{alg:arbitrary}} from visiting the same vertex in the same state more than once in the same tree. 
Consider the example in Figure \ref{fig:spanning_tree_before_insert}:
node $(u,2)$ is not added as a child of the node $(x,1)$ after traversing  edge $(x,u) \in S$ with label $\textit{mentions}$ since $(u,2)$ is already reachable from $(x,0)$.

Algorithm \textbf{\ref{alg:insertrapq_expiry}} is invoked at pre-defined slide intervals to remove expired nodes from $\Delta$.
For each $T_x \in \Delta$, it identifies the set of candidate nodes whose timestamps are not in  $(\tau - |W|, \tau ]$ (Line \ref{line:insertrapq_candidatenodes}) and temporarily removes those from $T_x$ (Line \ref{line:insertrapq_candidatesubgraph}).
For each candidate $(v,t)$, Algorithm \textbf{\ref{alg:insertrapq_insert_edge}} finds an incoming edge from another valid node in $T_x$ (Line \ref{line:insertrapq_connected}) and it reconnects the subtree rooted at $(v,t)$  to $T_x$. 
Nodes with no valid incoming edges are permanently removed from $T_x$.
Algorithm \textbf{\ref{alg:insertrapq_expiry}} might traverse the entire snapshot graph $G_{W,\tau}$ in the worst case.
This can be used to undo previously reported results if explicit window semantics is required (Line \ref{line:insertrapq_removeresult}), yet, we only do so to process explicit deletions as described in \S \ref{sec:arbitrary-delete}.

\begin{example}
Consider the example provided in Figure \ref{fig:spanning_tree_after_insert} and assume that window size is $15$ time units. 
Upon arrival of  edge $(w,u)$ with label $follows$ at $t=19$, nodes $(u,1)$ and $(x,2)$ are added to $T_x$ as descendants of $(w,2)$.
Also, paths leading to nodes $(u,2)$, $(v,1)$ and $(y,2)$ are  expired as their timestamp is $4$ (due to the edge $(y,u)$ with a timestamp $4$). 
Algorithm \textbf{\ref{alg:insertrapq_expiry}} searches incoming edges of  vertex $u$ in $G_{W, \tau}$ and identifies that there exists a valid edge $(z,u)$ with label $mentions$ and timestamp $14$.
As a result,  node $(u,2)$ and its subtree is reconnected to node $(z,1)$.
\end{example}


\begin{thm}\label{thm:arb_correctness}
Algorithm \textbf{\ref{alg:arbitrary}} is correct and complete.
\end{thm}

\begin{proof}
Algorithm \textbf{\ref{alg:arbitrary}} terminates as Line \ref{line:insertrapq_invariant2} ensures that no node is visited more than once in any spanning tree in $\Delta$.

\textbf{If:} If direction follows trivially from the first invariant of spanning trees.
Lemma \ref{lem:arbitraryinvariants} guarantees that node $(u,s)$ is inserted into the spanning tree $T_x$ if there exists a path in the snapshot graph $G_{W,\tau}$ of the window $W$ at time $\tau$ from $x$ to $u$ satisfying $R$. 
Line \ref{line:addresult} in Algorithm \textbf{\ref{alg:insertrapq_insert_edge}} adds the pair $(x,u)$ to the set of results $R$ if the target state is an accepting state, $s \in F$.

\textbf{Only If:} 
If the algorithm adds $(x,u)$ to $R$, then it must traverse a path $p$ from $x$ to $u$ in $G_{W,\tau}$ where $s_f=\delta^*(s_0,\phi(p)), s_f \in F$ and $ p.ts \in (\tau - |W|, \tau ]$. 
Let $n$ be the length of such path $p$.
For any $(x,u)$ that is added to $R$, Algorithm \textbf{\ref{alg:insertrapq_insert_edge}} must have been invoked with the node $(u,s_f)$ as the child node for some $s_f \in F$ (Line \ref{line:insertrapq_callinsert} in \textbf{\ref{alg:arbitrary}} or Line \ref{line:insertrapq_callinsert2} in \textbf{\ref{alg:insertrapq_insert_edge}}). 
Therefore, the proof proceeds by showing that  node $(u,s_f)$ with timestamp $(u,s_f).ts \in (\tau - |W|, \tau ]$ for some $s_f \in F$ is added to the spanning tree $T_x$ only if there exists a path $p$ of length $n$ with the same timestamp in $G_{W,\tau}$ from $x$ to $u$ satisfying $R$.
For the base case of $n = 1$, assume there exists a tuple $t_{\tau}=(\tau, e, l, +)$, $e=(x,u)$ where $\delta(s_0,l)=s_f$ for some $s_f \in F$. 
\textbf{Algorithm \ref{alg:arbitrary}} (Line \ref{line:insertrapq_callinsert}) invokes Algorithm \textbf{\ref{alg:insertrapq_insert_edge}} with parameters $(x,s_0)$ as the parent node and $(u,s_f)$ as the child node, then $(x,u)$ with timestamp $\tau$ is added to the result set (Line \ref{line:addresult}).
Let's assume that there exists a path $q$ of length $n-1$ in $G_{W,\tau}$ from  $x$ to $v$ where $t=\delta^*(s_0,\phi(p))$ and there exists a node $(v,t)$ in $T_x$ where $(v,t).ts = q.ts \in (\tau - |W|, \tau ]$.
For the node $(u,s)$ to be added to the spanning tree $T_x$ with timestamp $(u,s).ts \in (\tau - |W|, \tau ]$, Algorithm \textbf{\ref{alg:insertrapq_insert_edge}} must have been invoked with $(u,s)$ by Line \ref{line:insertrapq_callinsert} of Algorithm \textbf{\ref{alg:arbitrary}} or Line \ref{line:insertrapq_callinsert2} of Algorithm \textbf{\ref{alg:insertrapq_insert_edge}}.
In either case, there must be  an edge $e=(v,u) \in G_{W,\tau}$ where  $s=\delta(t,\phi(u,v))$, and $e.ts \in (\tau - |W|, \tau ]$.  Therefore, this implies that there exists a path of length $n$ in $G_{W,\tau}$ from  $x$ to $u$, thus concluding the proof.
\end{proof}


\begin{algorithm}
\small
\SetAlgoRefName{ExpiryRAPQ}
	\SetKwData{Left}{left}\SetKwData{This}{this}\SetKwData{Up}{up}
	\SetKwInOut{Input}{input}\SetKwInOut{Output}{output}
	\Input{Window $G_{W,\tau}$, timestamp $\tau$,  Spanning tree $T_x$
	}
	\Output{
	The set of invalidated results $R_I$
	}
	$R_I \leftarrow \emptyset$ \\
    set $P = \{ (v,t) \in T_x \mid (v,t).ts \leq \tau - |W| \}$ \tcp{potentially expired nodes} \label{line:insertrapq_candidatenodes} 
    $T_x \leftarrow T_x \setminus P$ \tcp{prune $T_x$}
    \label{line:insertrapq_candidatesubgraph} 
    \ForEach{$(v,t) \in  P$}{ \label{line:insertrapq_traversenodes}
        \ForEach{$(u,v) \in W_{G,\tau}$}{ \label{line:insertrapq_traverseedges}
            \If{$(u,s) \in T_x \wedge t=\delta(s, \phi(u,v))$}{ 
	            $P \leftarrow P \setminus$ \ref{alg:insertrapq_insert_edge}($T_x, (u,s), (v,t), (u,v) $) \label{line:insertrapq_connected} \\
            }
        }
    }
    
    \ForEach{$(v,t) \in P$}{ \label{line:insertrapq_candidateremaining}
        \If{$t \in F$}{
            $R \leftarrow R + (x,v)$ \label{line:insertrapq_removeresult}
        }
    }

	\Return $R_I$

\caption{}\label{alg:insertrapq_expiry}
\end{algorithm}

\begin{thm}\label{thm:rapq_insert_amortized}
The amortized cost of Algorithm \textbf{\ref{alg:arbitrary}} is $\mathcal{O}(n \cdot k^2)$, 
where $n$ is the number of distinct vertices in the window $W$ and $k$ is the number of states in the corresponding automaton $A$ of the the query $Q_R$. 
\end{thm}
\begin{proof}
Consider a tuple $t_\tau$ with an edge $e=(u,v)$ and label $l$ arriving for processing. 
Updating window $G_{W,\tau}$ with  edge $e$ (Line \ref{line:insertrapq_addwindow}) takes constant time. 
Thus, the time complexity of Algorithm \textbf{\ref{alg:arbitrary}} is the total number of times  Algorithm \textbf{\ref{alg:insertrapq_insert_edge}} is invoked.

First, we show that the amortized cost of updating a single spanning tree $T_x$ rooted at $(x,s_0)$ is constant in window size.
For an edge $(u,v)$ with label $l$, there could be $k$ many parent nodes $(u,s) \in T_x$ for each state $s$, and thus there could be at most $k^2$ invocations of Algorithm \textbf{\ref{alg:insertrapq_insert_edge}} with child node $(v,t)$, for each state $t$.
Upon arrival of the edge $e=(u,v)$, Algorithm \textbf{\ref{alg:insertrapq_insert_edge}} is invoked with nodes $(u,s)$ as parent and $(v,t)$ as child either when $(u,s)$ is already in $T_x$ at time $\tau$, $\tau - |W| < (u,s).ts \leq \tau$ (Line \ref{line:insertrapq_callinsert} in Algorithm \textbf{\ref{alg:arbitrary}}), or when $(u,s)$ is added to $T_x$ at a later point in time $(u,s).ts > \tau$ (Line \ref{line:insertrapq_callinsert2} in Algorithm \textbf{\ref{alg:insertrapq_insert_edge}}).
Note that Algorithm \textbf{\ref{alg:insertrapq_insert_edge}} is invoked with these parameters
at most once as Line \ref{line:insertrapq_invariant2} of Algorithm \textbf{\ref{alg:insertrapq_insert_edge}} extends a  node $(v,t)$ only if it is not in $T_x$.
The second invariant (Lemma \ref{lem:arbitraryinvariants}) guarantees that  $(u,s)$ appears in a spanning tree $T_x$ at most once.
Therefore, 
Algorithm \textbf{\ref{alg:insertrapq_insert_edge}} is invoked at most $m \cdot k^2$ over a sequence of $m$ tuples.
As there are at most $n$ spanning trees in $\Delta$, one for each $x \in G_{W,\tau}$, the total amortized cost 
is $\mathcal{O}(n \cdot k^2)$.
\end{proof}

Consequently, Algorithm {\textbf{\ref{alg:insertrapq_insert_edge}}} has {$\mathcal{O}$}$(n)$ amortized time complexity in terms of the number of vertices in the snapshot graph $G_{W,\tau}$.
As described previously,  Algorithm \textbf{\ref{alg:insertrapq_expiry}} might traverse the entire product graph 
and its worst case  complexity is $\mathcal{O}(m \cdot k^2)$. 
Therefore, the total cost of window maintenance over $n$ spanning trees is $\mathcal{O}(n \cdot m \cdot k^2)$.
This cost is amortized over the window slide interval $\beta$.

\subsection{Explicit Deletions}
\label{sec:arbitrary-delete}

The majority of real-world applications process append-only streaming graphs where existing tuples in the window expire only due to window movements.
However, there are applications that require users to explicitly delete a previously inserted edge.
We show that Algorithm \textbf{\ref{alg:insertrapq_expiry}} proposed in \S \ref{sec:arbitrary-insertonly} can be utilized  to support such explicit edge deletions.
Remember that in the append-only case, a node $(v,t)$ in a spanning tree $T_x \in \Delta$ is only removed when its timestamp falls outside the window range. 
An explicit deletion might require $(v,t) \in T_x$ to be removed if the deleted edge is on the path from $(x,s_0)$ to $(v,t)$ in the spanning tree $T_x$.
We utilize Algorithm \textbf{\ref{alg:insertrapq_expiry}} to remove such nodes so that explicit deletions and window management are handled in a uniform manner.

\begin{dft}[Tree Edge]\label{def:tree-edge}
Given a spanning tree $T_x$ at time $\tau$, an edge $e=(u,v)$ with label $l$ is a \textit{tree-edge} w.r.t $T_x$ if $(u,s)$ is the parent of $(v,t)$ in $T_x$ and there is a transition from state $s$ to $t$ with label $l$, i.e., $(u,s) \in T_x$, $(v,t) \in T_x$, $t=\delta(s,l)$, and $(v,t).pt = (u,s)$.
\end{dft}

Algorithm \textbf{\ref{alg:rapq_delete_edge}}  finds spanning trees where a deleted edge $(u,v)$ is a tree-edge (Line \ref{line:deleterapq_treeedge}) as per Definition \ref{def:tree-edge}. 
Deletion of the tree-edge from $(u,s)$ to $(v,t)$ in $T_x$ disconnects $(v,t)$ and its descendants from $T_x$. 
Algorithm \textbf{\ref{alg:rapq_delete_edge}} traverses the subtree rooted at $(v,t)$ and sets the timestamp of each node to $-\infty$, essentially marking them as expired (Line \ref{line:deleterapq_settimestamp}).
Algorithm \textbf{\ref{alg:insertrapq_expiry}} processes each expired node in $\Delta$ and checks if there exists an alternative path  comprised of valid edges in the window.
Algorithm \textbf{\ref{alg:rapq_delete_edge}} invokes  Algorithm \textbf{\ref{alg:insertrapq_expiry}} (Line \ref{line:deleterapq_callexpiry}) to manage explicit deletions using the same machinery of window management.
Deletion of a non-tree edge, on the other hand, leaves spanning trees unchanged so no modification is necessary other than updating the window content $G_{W,\tau}$.






	


\begin{algorithm}
\small
\SetAlgoRefName{Delete}
\caption{}\label{alg:rapq_delete_edge}
	\SetKwData{Left}{left}\SetKwData{This}{this}\SetKwData{Up}{up}
	\SetKwInOut{Input}{input}\SetKwInOut{Output}{output}
	\Input{Incoming tuple  $t_{\tau} = (\tau,e_{\tau},l,-), e_{\tau}=(u,v)$, \\
	    Window $G_{W,\tau - 1}$
	}
	\Output{
	The set of invalidated results $R_I$
	}
    $R_I \leftarrow \emptyset$ \\
    
	\ForEach{$T_x \in \Delta$}{
    	\ForEach{$s,t \in S \mid t=\delta(s,l) \land (v,t) \in T_x \land (v,t).pt = (u,s)$\label{line:deleterapq_treeedge}}{ 
    	    $T_{(x,v,t)} \leftarrow$ the subtree of $(v,t)$ in $T_x$ \label{line:deleterapq_traverse} \\
    	    \ForEach{$(w,q) \in T_{(x,v,t)}$\label{line:deleterapq_settimestamp}}{
    	        $(w,q).ts = -\infty$
    	    }
    	}
$R_I \leftarrow R_I \cup \ref{alg:insertrapq_expiry}(W_{G, \tau}, T_x, \tau)$ \label{line:deleterapq_callexpiry} \\
}
\Return  $R_I$ 
\end{algorithm}

\begin{thm}\label{thm:rapq_delete_amortized}
The amortized cost of Algorithm \textbf{\ref{alg:rapq_delete_edge}} is $\mathcal{O}(n^2 \cdot k)$
over a sequence of explicit edge deletions.
\end{thm}
\begin{proof}
First, we evaluate the cost of an explicit deletion over a single spanning tree $T_x \in \Delta$, rooted at $(x,s_0)$.
Given a negative tuple with edge $(u,v)$ and label $l$, Line \ref{line:deleterapq_treeedge} identifies the corresponding set of tree edges in $T_x$ in $\mathcal{O}(n \cdot k)$ time.
For each such tree edge from $(u,s)$ to $(v,t)$ in $T_x$, Line \ref{line:deleterapq_traverse} traverses the spanning tree $T_x$ starting from $(v,t)$ to identify the set of nodes that are possibly affected by the deleted edge, thus its cost is $\mathcal{O}(n \cdot k)$.
Once timestamps of nodes in the subtree of $(v,t)$ is set to $-\infty$, Line \ref{line:deleterapq_callexpiry} invokes Algorithm \textbf{\ref{alg:insertrapq_expiry}} to process all expired nodes in $T_x$, whose time complexity is $\mathcal{O}(m \cdot k^2)$.
There can be at most $m \cdot k^2$ edges in the product graph of  snapshot $G_{W,\tau}$ with $m$ edges and automaton $A$ with $k$ edges. 
The amortized time complexity of maintaining a single spanning tree $T_x \in \Delta$ over a sequence of $m$ explicit deletion is $\mathcal{O}(n \cdot k)$ since at most $n \cdot k$ of those edges are tree edges.
Algorithm \textbf{\ref{alg:rapq_delete_edge}} does not need to process non-tree edges as a removal of a non-tree edge only need to update the window $G_{W,\tau}$, which is a constant time operation.
Therefore, the amortized cost of Algorithm \textbf{\ref{alg:rapq_delete_edge}} over a sequence of $m$ explicit edge deletions is $\mathcal{O}(n^2 \cdot k)$.
\end{proof}


\section{RPQ with Simple  Path Semantics}
\label{sec:simple}

In this section, we turn our attention to the problem of persistent RPQ evaluation on streaming graphs under the simple path semantics, that is finding pairs of vertices $u,v \in V$ where there exists a simple path (no repeating vertices) $p$ between $u$ and $v$ with a path label $w$ in the language $L(R)$.

The decision problem for Regular Simple Path Query (RSPQ), i.e., deciding whether a pair of vertices $u,v \in V$ is in the result set of a RSPQ $Q_R$, is NP-complete for certain fixed regular expressions, making the general problem NP-hard \cite{mendelzon1995finding}.
Mendelzon and Wood  \cite{mendelzon1995finding} 
show that there exists a batch algorithm to evaluate RSPQs on static graphs in the absence of conflicts, a condition on the cyclic structure of the graph $G$ and  the regular language $L(R)$ of the query $Q_R$.

\begin{dft}[Suffix Language]
Given an automaton $A = (S, \Sigma, \delta, s_0, F)$, 
the suffix language of a state $s$ is defined as $ [s] = \{ w \in \Sigma^* \mid \delta^*(s,w) \in F \}$; that is, the set of all strings that take $A$ from state $s$ to a final state $s_f \in F$. 
\end{dft}

\begin{dft}[Containment Property]
Automaton $A = (S, \Sigma, \delta, s_0, F)$ has the \textit{suffix language containment property} if for each pair $(s,t) \in S\times S$ such that $s$ and $t$ are on a path from $s_0$ to some final state and $t$ is a successor of $s$, $[s] \supseteq [t]$.
\label{def:containment}
\end{dft}

We compute and store the suffix language containment relation for all pairs of states during query registration, 
i.e., the time when the query $Q_R$ is first posed, and use these in the proposed streaming algorithm to detect conflicts. 
We can now precisely define conflicts.

\begin{dft}[Conflict]
\label{def:conflict}
There is a conflict at a vertex $u$  if and only if a traversal of the product graph $P_{G,A}$ starting from an initial node $(x,s_0) \in P_{G,A}$  visit node $u$ in states $s$ and $t$, and $[s] \not\supseteq [t]$. In other words, a tree $T_X$ is said to have a conflict between states $s$ and $t$ at vertex $u$ if $(u,s)$ is an ancestor of $(u,t)$ in the spanning tree $T_x$ and $[s] \not\supseteq [t]$.
\end{dft}

\begin{example}
Consider the streaming graph and the query in Figure {\ref{fig:running_example}} and the its spanning tree given in Figure {\ref{fig:spanning_tree_before_insert}}.
The node $(y,2)$ is added as a child of the node $(v,1)$ when edge $(v,y)$ arrives at $t=18$.
Based on Definition {\ref{def:conflict}}, there is a conflict at vertex $v$ as the path $p$ from the root node $(x,0)$ visits the vertex $v$ at states $1$ and $2$, and $[1] \not\supseteq [2]$.
\end{example}

\textbf{Batch Algorithm}:
Similar to the batch algorithm in \S \ref{sec:arbitrary}, 
the batch RSPQ algorithm \cite{mendelzon1995finding} starts a DFS traversal of the product graph from every vertex $x \in V$ with the start state $s_0$, and constructs a DFS tree, $T_x$. 
Each DFS tree maintains a set of markings that is used to prevent a vertex being visited more than once in the same state in a $T_x$. 
A node $(u,s)$ is added to the set of markings only if the depth-first traversal starting from the node $(u,s)$ is completed and no conflict is detected.
Mendelzon and Wood \cite{mendelzon1995finding} show that a RSPQ $Q_R$  can be evaluated in 
{$\mathcal{O}$}$(n \cdot m )$
in terms of the size of the graph $G$  by the batch algorithm in the absence of conflicts -- 
the same as the batch algorithm for RAPQ evaluation presented in \S \ref{sec:arbitrary}.
A query $Q_R$  on a graph $G$ is conflict-free if: 
(i) the automaton $A$ of $R$ has the suffix language containment property, (ii) $G$ is an acyclic graph, or (iii) $G$ complies with a cycle constraint compatible with $R$.
In following, 
we study the persistent RSPQ evaluation problem and show that the notion of \textit{conflict-freedom} ~\cite{mendelzon1995finding} is applicable to sliding windows over streaming graphs, admitting an efficient evaluation algorithm in the absence of conflicts.

\subsection{Append-only Streams}
\label{sec:simple-insertonly}
First, 
we present an incremental algorithm for RSPQ evaluation based on its RAPQ counterpart (Algorithm \textbf{\ref{alg:ProcessEdgeRSPQ}}) with implicit window semantics and we show that the proposed streaming algorithm matches the complexity characteristics of the batch algorithm for RSPQ evaluation on static graphs \cite{mendelzon1995finding}, i.e., it admits efficient solutions under the same conditions as the batch algorithm.


\begin{algorithm}
\small
\SetAlgoRefName{RSPQ}
	\SetKwData{Left}{left}\SetKwData{This}{this}\SetKwData{Up}{up}
	\SetKwInOut{Input}{input}\SetKwInOut{Output}{output}
	\Input{Incoming tuple $t_{\tau} = (\tau,e_{\tau},l,op), e_{\tau}=(u,v)$
	        }
	$G_{W,\tau} \leftarrow G_{W,\tau-1} + e_{\tau}$  \label{line:insertrspq_addwindow} \\
     \ref{alg:insertrspq_expiry}($G_{W, \tau}, T_x, \tau$) $\forall T_x \in \Delta$ \label{line:insertrspq_callexpiry} \tcp{ with $\beta$ intervals}
    set of results $R \leftarrow \emptyset $ 


    \ForEach{$T_x \in \Delta$}{ \label{line:insertrspq_foreachtree}
        \ForEach{$s,t \in S$ where $t=\delta(s,l)$}{ \label{line:insertrspq_foreachtransition}
            \If{$(u,s) \in T_x \wedge (u,s).ts > \tau - |W| $}{ \label{line:insertrspq_windowedge}
                $p \leftarrow PATH(T_x, (u,s))$ \tcp{the prefix path} \label{line:insertrspq_retrievepp} 
                \If{$ t \not\in p[v] \wedge (v,t) \not\in M_x$}{ \label{line:epp_case1_1}
                    $R \leftarrow R +$ \ref{alg:extend}($T_x, p, (v,t), e_{\tau}$) \label{line:deleterspq_invokeepp1} 
                }
            }
        }
    }
    
	$Q_R(S, W, \tau) \leftarrow Q_R(S, W, \tau-1) + R$ \label{line:rspq_addresult}
	\caption{}\label{alg:ProcessEdgeRSPQ}
\end{algorithm} 

\begin{algorithm}
\small
\SetAlgoRefName{Extend}
	\SetKwData{Left}{left}\SetKwData{This}{this}\SetKwData{Up}{up}
	\SetKwInOut{Input}{input}\SetKwInOut{Output}{output}
	\Input{Spanning Tree $T_x$, Prefix Path $p$, \\ 
	Node $(v,t)$, Edge $e=(u,v)$ \\
	}
	\Output{ Set of results $R$
	}
	$R \leftarrow \emptyset$ \\
    \uIf{$q = FIRST(p[v])$ and $[q] \not\supseteq [t]$}{ \label{line:epp_case3}
        \ref{alg:unmark}($T_x, p$) \tcp{$q$ and $t$ have a conflict at vertex $v$} 
    }

    \Else{ \label{line:epp_case4}
        \If{$t \in F$}{
            $R \leftarrow  R + (x,v)$ \label{line:epp_addresult}
        }
        \If{$(v,t) \notin T_x$}{
            $M_x \leftarrow M_x \bigcup (v,t)$ \label{line:deleterspq_addmarking}\\
        }
        add $(v,t)$ as $(u,s)$'s child in $T_x$ 
        \label{line:deleterspq_addleaf}\\
        $p_{new} \leftarrow p + [v,t]$\\
        $p_{new}.ts = min(e.ts, p.ts)$ \\
        \ForEach{edge $e=(v,w) \in W_{G,\tau}$ s.t. $\delta(t, \phi(e))=r$} {    
            \If{$ r \not\in p_{new}[w] \wedge (w,r) \not\in M_x$}{\label{line:epp_case1_2}
                $R \leftarrow R +$ \ref{alg:extend}($T_x, p_{new}$, $(w,r), e$)  \label{line:deleterspq_invokeepp2} \\ 
            }
        }
    }
    \Return $R$
	\caption{}\label{alg:extend}
\end{algorithm}

\begin{algorithm}
\small
\SetAlgoRefName{Unmark}
	\SetKwData{Left}{left}\SetKwData{This}{this}\SetKwData{Up}{up}
	\SetKwInOut{Input}{input}\SetKwInOut{Output}{output}
	\Input{Spanning Tree $T_x$, Prefix Path $p$}
    $Q \leftarrow \emptyset$  \\
    \While{$p \neq \emptyset \wedge (v,t) = LAST(p) \wedge (v,t) \in M_x $}{
        $M_x \leftarrow M_X \setminus (v,t)$ \\
        $Q \leftarrow Q + (v,t) $ \\
        $p \leftarrow PATH(T_x, (v,t).parent$
    }
    
    \ForEach{$(v,t) \in Q$}{
        \ForEach{edge $e=(w,v) \in G_{W,\tau}$ s.t. $ t=\delta(q, \phi(e))$ \label{line:unmark_incoming}}{
            \If{$(w,q) \in T_x \wedge t \notin p[v]$ }{
                $p_{candidate} \leftarrow PATH(T_x, (w,q))$\\
                \ref{alg:extend}($T_x, p_{candidate}, (v,t), e$) \\
            }
        }
	}
	\caption{}\label{alg:unmark}
\end{algorithm}

\begin{algorithm}
\small
\SetAlgoRefName{ExpiryRSPQ}
	\SetKwData{Left}{left}\SetKwData{This}{this}\SetKwData{Up}{up}
	\SetKwInOut{Input}{input}\SetKwInOut{Output}{output}
	\Input{Window $G_{W,\tau}$, timestamp $\tau$, \\
	        Spanning Tree $T_x$
	}
	\Output{The set of invalidated results $R_I$}
   
    $R_I \leftarrow \emptyset$ \\
    $E = \{ (v,t) \in T_x \mid (v,t).ts \leq \tau - |W| \}$ \tcp{expired nodes} \label{line:insertrspq_candidatenodes} 
    $P \leftarrow M_x \cap E$
    
    $T_x \leftarrow T_x \setminus E$ \tcp{prune $T_x$} \label{line:insertrspq_prunetree} 
    $M_x \leftarrow M_x \setminus E$ \tcp{prune $M_x$} \label{line:insertrspq_prunemarking} 
    
    \ForEach{$(v,t) \in  P$ }{ \label{line:insertrspq_traversenodes}
        \ForEach{$(u,v) \in W_{G,\tau}$ s.t. $(u,s) \in T_x \wedge t=\delta(s, \phi(u,v))$}{ \label{line:insertrspq_traverseedges}
            $p \leftarrow PATH(T_x, (u,s))$ \\
            $P \leftarrow P \setminus $ \ref{alg:extend}($T_x, p, (v,t), (u,v)$)
            }
        
    }
    
    \ForEach{$(w,q) \in P$ \label{line:deleterspq_loop2}}{

        \If{all siblings of $(w,q)$ are in $M_x$}{
            $M_x \leftarrow M_x + (w,q).parent$
        }
        \If{$q \in F$}{
            $R_I \leftarrow R_I + (x,w)$ \label{line:deleterspq_removeresult}
        }
    }
    \Return $R_I$
	\caption{}\label{alg:insertrspq_expiry}
\end{algorithm}


\begin{dft}[Prefix Paths]\label{def:prefixpath}
Given a node $(u,s) \in T_x$, we say that the path from the root to $(u,s)$ is the prefix path $p$ for node $(u,s)$. 
We use the notation $p[v], v\in V$ to denote the set of states that are visited in vertex $v$ in path $p$, i.e., $p[v] = \{ s \in S \mid (v,s) \in p\}$.
\end{dft}

\begin{dft}[Conflict Predecessor]\label{def:conflict_predecessor}
A node $(u,s) \in T_x$ is a conflict predecessor if for some successor $(w,t)$ of $(u,s)$ in $T_x$, $(w,q)$ is the first occurrence of vertex $w$ in the prefix path of $(u,s)$ and there is a conflict between $q$ and $t$ at $w$, i.e., $[q] \not\supseteq [t]$.
\end{dft}

In addition to tree index $\Delta$ of Algorithm \textbf{\ref{alg:arbitrary}} in \S \ref{sec:arbitrary}, Algorithm \textbf{\ref{alg:ProcessEdgeRSPQ}} maintains a set of markings $M_x$ for each spanning tree $T_x$. 
The set of markings $M_x$ for a spanning tree $T_x$ is the set of nodes in $T_x$ with no descendants that are conflict predecessors (Definition {\ref{def:conflict_predecessor}}). 
In the absence of conflicts, there is no conflict predecessor and 
$M_x$ contains all nodes in $T_x$.
Algorithm \textbf{\ref{alg:ProcessEdgeRSPQ}} does not visit a node in $M_x$ (Lines \ref{line:epp_case1_1} in Algorithm \textbf{\ref{alg:ProcessEdgeRSPQ}} and \ref{line:epp_case1_2} in Algorithm \textbf{\ref{alg:extend}}) and therefore
a node $(u,s)$ appears in the spanning tree $T_x$ at most once in the absence of conflicts.
Consequently, Algorithm \textbf{\ref{alg:ProcessEdgeRSPQ}} maintains the second invariant of $\Delta$ and behaves similar to the  Algorithm \textbf{\ref{alg:arbitrary}} presented in \S \ref{sec:arbitrary-insertonly}.
On static graphs, the batch algorithm adds a node $(u,s)$ to the set of markings only after the entire depth-first traversal of the product graph from $(u,s)$ is completed, ensuring that the set $M_x$ is monotonically growing. 
On the other hand, tuples that arrive later in the streaming graph $S$ might lead to a conflict with a node  $(u,s)$ that is already in $M_x$, and Algorithm \textbf{\ref{alg:ProcessEdgeRSPQ}} removes $(u,s)$'s ancestors from the set of markings $M_x$.
As described later, Algorithm \textbf{\ref{alg:ProcessEdgeRSPQ}} correctly identifies these conflicts  and updates the spanning tree $T_x$ and its set of markings $M_x$ to ensure  correctness.
The conflict detection mechanism signals to our algorithm that the corresponding traversal cannot be pruned even if it visits a previously visited vertex. 
In other words, a node $(u,s) \not\in M_x$ may be visited more than once in a spanning tree $T_x$ to ensure correctness.
Consequently, Algorithm \textbf{\ref{alg:ProcessEdgeRSPQ}}  traverses every simple path that satisfies the given query $Q_R$ if every node in $T_x$ is a conflict predecessor ($M_x = \emptyset$), leading to exponential time execution in the worst case.
In summary, Algorithm \textbf{\ref{alg:ProcessEdgeRSPQ}} differs from its arbitrary path semantics counterpart in two major points: (i) it may traverse a vertex in the same state more than once if a conflict is discovered at the vertex, and (ii) it keeps track of conflicts and maintains a set of markings to prevent multiple visits of the same vertex in the same state whenever possible.

For each incoming tuple $t_{\tau} (ts,e,l,+),e=(u,v)$, Algorithm \textbf{\ref{alg:ProcessEdgeRSPQ}} finds prefix paths of all $(u,s) \in T_x$ (Line \ref{line:insertrspq_retrievepp} ); that is, the set of  paths in $T_x$ from the root node to $(u,s)$  (note that there exists a single such node $(u,s)$ and its corresponding prefix path if $(u,s) \in M_x$).
Then it performs one of the following four steps for each node $(u,s) \in T_x$ and its corresponding prefix path $p$:
\begin{enumerate}
    \item $t \in p[v]$: The vertex $v$ is visited in the same state $t$ as before, thus path $p$ is pruned as extending it with $(v,t)$ leads to a cycle in the product graph $P_{G,A}$ (Line \ref{line:epp_case1_1} in \textbf{\ref{alg:ProcessEdgeRSPQ}} and Line \ref{line:epp_case1_2} \textbf{\ref{alg:extend}}).
    \item $(v,t) \in M_x$: The target node $(v,t)$ has already been visited in $T_x$ and it has no conflict predecessor descendant. Therefore  path $p$ is pruned (Line \ref{line:epp_case1_1} in \textbf{\ref{alg:ProcessEdgeRSPQ}}, \ref{line:epp_case1_2} in \textbf{\ref{alg:extend}}).
    \item $q = FIRST(p[v])$ and $[q] \not\supseteq [t]$: States $q$ and $t$ have a conflict at vertex $v$ (Line \ref{line:epp_case3} in \textbf{\ref{alg:extend}}), making $(u,s)$ a conflict predecessor.
    Therefore, all ancestors of $(u,s)$ in $T_x$ are removed from $M_x$ (Algorithm \textbf{\ref{alg:unmark}}).
    During unmarking of a node $(v_i, s_i) \in M_x$, all $(w,q) \in T_x$ where $(w,v_i) \in G_{W,\tau}$ and $s_i = \delta(q, \phi(w,v_i))$ are considered as candidate for traversal as they were previously pruned due to $(v_i, s_i)$ being marked. 
    \item Otherwise path $p$ is extended with $(v,t)$, i.e., $(v,t)$ is added as a child to $(u,s)$ in $T_x$. (Line \ref{line:epp_case4} in \textbf{\ref{alg:extend}})
\end{enumerate}

As described previously, an important difference between the proposed streaming algorithm and the batch algorithm \cite{mendelzon1995finding} is that the  streaming version may remove nodes from the set of markings $M_x$ whereas a node in $M_x$ cannot be removed in the batch model.
Hence, the batch algorithm can safely prune a path $p$ if it reaches a node $(u,s) \in M_x$ as the suffix language containment property ensures correctness.
The streaming model, on the other hand, requires a special treatment as $M_x$ is not monotonically growing.
Case 2  above prunes a path $p$ if it reaches a node $(u,s) \in M_x$ 
as in the batch algorithm.
Unlike the batch algorithm,  a node $(u,s)$ may be removed from $M_x$  due to a conflict that is caused by an edge that later arrives.
This conflict implies that  path $p$ should not have been pruned.
Case 3 above and Algorithm \textbf{\ref{alg:unmark}} address exactly this scenario: ancestors of a conflict predecessor is removed from $M_x$.

Whenever a node $(u,s)$ is removed from $M_x$ due to a conflict at one of its descendants, Algorithm \textbf{\ref{alg:unmark}} finds all paths that are previously pruned due to $(u,s)$ by traversing incoming edges of $(u,s) \in G_{W,\tau}$ and invokes Algorithm  \textbf{\ref{alg:extend}} for each such path.
It enables Algorithm  \textbf{\ref{alg:extend}} to backtrack and evaluate all paths that would not be pruned by Case 2 if $(u,s)$ were not in $M_x$, ensuring the correctness of the algorithm.


The following example illustrates this behaviour of Algorithm  \textbf{\ref{alg:ProcessEdgeRSPQ}}.

\begin{figure}
    \centering
    \includegraphics[height=1.6in]{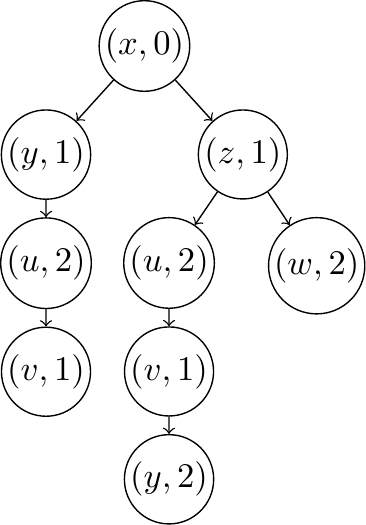}
    \caption{A spanning tree $T_x$ constructed by Algorithm \ref{alg:ProcessEdgeRSPQ} for the example  in Figure \ref{fig:running_example}.}
    \label{fig:spanning_tree_simple}
\end{figure}

\begin{example}
Consider the streaming graph and the query in Figure \ref{fig:running_example} and the its spanning tree given in Figure \ref{fig:spanning_tree_before_insert}, and assume for now that Algorithm  \textbf{\ref{alg:ProcessEdgeRSPQ}} does not detect conflicts and only traverses simple paths in $G_{W,\tau}$.
After processing  edge $(x,y)$ at time $t=13$, it adds node $(u,2)$ as a successor of $(y,1)$.
Edge $(z,u)$ arrives at $t=14$, however $(u,2)$ is not added as $(z,1)$'s child as $(u,2)$ already exists in $T_x$.
Later at $t=18$, edge $(v,y)$ arrives, but $(y,2)$ is not added  to the spanning tree $T_x$ as the path $\langle x,y,u,v,y \rangle$ forms a cycle in $G_{W,\tau}$.
As a result, $(y,2)$ is never visited and $(x,y)$ is never reported even though there exists a simple path in $G_{W,\tau}$
from $x$ to $y$, that is $\langle x,z,u,v,y \rangle$.

Instead, Algorithm  \textbf{\ref{alg:ProcessEdgeRSPQ}} detects the conflict at the vertex $v$ between states $1$ and $2$ after  edge $(v,y)$ arrives at time $t=18$ as $FIRST(p[y])=1$ and $[1] \not\supseteq [2]$.
Algorithm  \textbf{\ref{alg:unmark}} removes all ancestors of $(y,2)$ from $M_x$ and, during unmarking of $(u,2)$, the prefix path $p$ from $(x,0)$ to $(z,1)$ is extended with $(u,2)$.
Finally, Algorithm  \textbf{\ref{alg:extend}} traverses the simple path $\langle x,z,u,v,y \rangle$ and adds $(x,y)$ to the result set.
Figure \ref{fig:spanning_tree_simple} depicts the spanning tree $T_x \in \Delta$ at time $t=18$.
\end{example}

Similar to its arbitrary counterpart, Algorithm {\textbf{\ref{alg:ProcessEdgeRSPQ}}} invokes Algorithm {\textbf{\ref{alg:insertrspq_expiry}}} at each user-defined slide interval $\beta$.
It first identifies the set of candidate nodes whose timestamp is not in  $(\tau - |W|, \tau ]$ (Line {\ref{line:insertrspq_candidatenodes}}).
Unmarked candidate nodes ($M_x \setminus E$) can safely be removed from $T_x$ as the unmarking procedure already considers all valid edges to an unmarked node.
Hence, Algorithm {\textbf{\ref{alg:insertrspq_expiry}}} reconnects a candidate node with a valid edge only if it is marked (Line {\ref{line:insertrspq_traversenodes}}).
Finally, it extends the set of marking with nodes that are not conflict predecessors any longer (Line {\ref{line:deleterspq_loop2}}).

\begin{thm}\label{thm:insertrspq_correctness}
The algorithm \textbf{\ref{alg:ProcessEdgeRSPQ}} is correct and complete.
\end{thm}

\begin{proof}

\textbf{If:} If the proposed algorithm traverses the path $p$, it correctly adds it to the result set $R$ and consecutively $Q_R(G_{\tau})$ (Line \ref{line:epp_addresult} and \ref{line:rspq_addresult} in  Algorithm \textbf{\ref{alg:extend}}). 
The reason $p$ is not traversed is due to a marked node (Case 2 of the proposed algorithm) as no vertex appears more than once in $p$ (as it is a simple path).
Let the last node visited in $p$ be $(v,t)$ and its successor on $p$ be $(w,r)$.
The initial part of path $p$ from $(x,s_0)$ to $(v,t)$ is not extended by $(w,r)$ as $(w,r) \in M_x$
If $(w,r)$ is removed from $M_x$ due to a conflict predecessor descendant of $(w,r)$, Algorithm \textbf{\ref{alg:unmark}} guarantees that the initial part of path $p$ from $(x,s_0)$ to $(v,t)$ 
is extended with $(w,r)$ as $(v,t) \in T_x$ and $(v,w) \in E$ and $r=\delta(t,\phi(v,w))$  (Line \ref{line:epp_case3} of Algorithm \textbf{\ref{alg:unmark}}). 
As a result, the path from $(v,t)$ to $(u,s_f)$ is discovered and $(x,u)$ is added to $Q_R(G_{\tau})$.
If $(w,r)$ remains in $M_x$, we know that $(w,r)$ does not have any descendants that is a conflict predecessor.
Therefore, $(u,s)$ must have been traversed as a descendant of $(w,r)$, adding $(x,u)$ to $Q_R(G_{\tau})$.

\textbf{Only if:} Assume that $p$ is not simple, meaning that there exists a node $v$ that appears in $p$ more than once. 
The first such occurrence is $(v,s_1) \in p$ and the last such occurrence is $(v,s_2) \in p$. 
For $(v,s_2)$ to be visited, $[s_1] \not\supseteq [s_2]$ must have been false (Line \ref{line:epp_case3} in Algorithm \textbf{\ref{alg:extend}}). 
The containment property (Definition \ref{def:containment}) implies that there exists a path $p'$ from $(v,s_1)$ to $(u,s_f^2)$, $s_f^2 \in F$ such that the sequence of vertices on $p'$ is identical to those in $p$ from $(v,s_2)$ to $(u,s_f)$. 
Note that $(v,s_1)$ and $(v,s_2)$ are the first and last occurrences of $v$ in $p$, therefore there exists a simple path in $P_{G,A}$ from $(x,s_0)$ to $(u,s_f^2), s_f^2 \in F$ where the vertex $v$ appears only once.
By simple induction on the number of repeated vertices, we conclude that there is a simple path in $G$ from $x$ to $u$ where the path label is in $L(R)$, and thus $(x,u)$ is added to $Q_R(G_{\tau})$. 
\end{proof}

\begin{thm}\label{thm:rspq_insert_amortized}
The amortized cost of Algorithm \textbf{\ref{alg:ProcessEdgeRSPQ}} is $\mathcal{O}(n \cdot k^2)$,
where $n$ is the number of distinct vertices in the window $W$ and $k$ is the number of states in the corresponding automaton $A$ of the query $Q_R$.
\end{thm}

\begin{proof}
It is important to stress that the proposed algorithm might take exponential time in the size of the stream in the presence of conflicts as RSPQ evaluation is NP-hard in its general form \cite{mendelzon1995finding}.
Therefore, first we focus on streaming RSPQ evaluation in the absence of conflicts and show that the cost of updating a single spanning tree $T_x$ and its markings $M_x$ is constant in the size of the stream.

The cost of Algorithm \textbf{\ref{alg:ProcessEdgeRSPQ}} for updating a single spanning tree $T_x$ is determined by the total cost of invocations of Algorithm \textbf{\ref{alg:extend}}.
In the absence of conflicts, Algorithm \textbf{\ref{alg:extend}} never invokes Algorithm \textbf{\ref{alg:unmark}}, and the cost of updating $R$ (Line \ref{line:epp_addresult}), $M_x$ (Line \ref{line:deleterspq_addmarking}) and $T_x$ (Line \ref{line:deleterspq_addleaf}) are all constant. 
Therefore the cost of Algorithm \textbf{\ref{alg:extend}} and thus the cost of Algorithm \textbf{\ref{alg:ProcessEdgeRSPQ}} are determined by the number of invocations of Algorithm \textbf{\ref{alg:extend}}.

Algorithm \textbf{\ref{alg:extend}} checks if a prefix path $p$ whose last node in $(u,s)$ for some $t=\delta(s,l)$ can be extended with $(v,t)$.
We argue that each node $(v,t)$ appears in $T_x$ at most once.
The first time Algorithm \textbf{\ref{alg:extend}} is invoked with some prefix path $p$ and node $(v,t)$, path $p$ is extended and node $(v,t)$ is added to $T_x$ and $M_x$ (Line \ref{line:epp_case4}).
Consecutive invocation of Algorithm \textbf{\ref{alg:extend}} with node $(v,t)$ does not perform any modifications on $T_x$ or $M_x$
as $(v,t)$ is guaranteed to remain marked in absence of conflicts. 
Therefore, each node $(v,t)$ appears only once in each spanning tree $T_x$ in the absence of conflicts (a node is removed from $M_x$ only if a conflict is discovered at Line \ref{line:epp_case3}).
For an incoming tuple with edge $(u,v)$ with label $l$, there can be at most $k^2$ pairs of prefix path $p$ of $(u,s)$ and node $(v,t)$, for each $s,t \in S$.
Algorithm \textbf{\ref{alg:extend}} is invoked for each such pair at most once; either (i) when the edge $e=(u,v)$ first appears in the stream and $(u,s) \in T_x$ but not $(v,t)$ (Line \ref{line:deleterspq_invokeepp1}), or (ii) $e=(u,v)$ with label $l$ already appeared in the stream when $(u,s)$ is first added to $T_x$ and $(v,t) \notin T_x$ (Line \ref{line:deleterspq_invokeepp2}).
Over a stream of $m$ tuples, Algorithm \textbf{\ref{alg:extend}} is invoked $\mathcal{O}(m \cdot k^2)$ times for the maintenance of a spanning tree $T_x$. Therefore, amortized cost of maintaining a spanning tree $T_x$ over a stream of $m$ edges is $\mathcal{O}(k^2)$.
Given that there are $\mathcal{O}(n)$ spanning trees, one for each $x \in V$, the amortized complexity of Algorithm \textbf{\ref{alg:ProcessEdgeRSPQ}} is $\mathcal{O}(n \cdot k^2)$ per tuple. 

\end{proof}

Consequently, the amortized cost of Algorithm {\textbf{\ref{alg:ProcessEdgeRSPQ}}} is linear in the number $n$ of vertices in the snapshot graph $G_{W,\tau}$, similarly to its RAPQ counterpart (described in {\S} {\ref{sec:arbitrary-delete}}).
The algorithm  \textbf{\ref{alg:ProcessEdgeRSPQ}} processes explicit deletions  in the same manner as its RAPQ counterpart (described in \S \ref{sec:arbitrary-delete}).
Similarly, the amortized cost of processing sequence of $m$ explicit deletions is $\mathcal{O}(n^2 \cdot k)$ in the absence of conflicts, where $n$ is the number of distinct vertices and $k$ is the number of states in the corresponding automaton of a RSPQ $Q_R$.





\section{Experimental Analysis}
\label{sec:experimens}

We study the feasibility of the proposed persistent RPQ evaluation algorithms on both real-world and synthetic streaming graphs. 
We first systematically evaluate the throughput and the edge processing latency of Algorithm \textbf{\ref{alg:arbitrary}} on append-only streaming graphs, and  analyze the factors affecting its performance (\S \ref{sec:experiments_analysis}).
Then, we assess its scalability by varying the window size $|W|$, the slide interval $\beta$ and the query size $|Q_R|$ (\S \ref{sec:experiments_scalability}).
The overhead of Algorithm \textbf{\ref{alg:rapq_delete_edge}} over Algorithm \textbf{\ref{alg:arbitrary}} for explicit deletions is analyzed in \S \ref{sec:experiments_deletion}  whereas \S \ref{sec:experiments_simplepath} analyzes the feasibility of \textbf{\ref{alg:ProcessEdgeRSPQ}} for persistent RPQ evaluation under simple path semantics.
Finally we compare our proposed algorithms with other systems (\S\ref{sec:experiments_comparison}). 
Since this the first work to address RPQ evaluation over streaming graphs, we perform this comparison with respect to an emulation of persistent RPQ evaluation on RDF systems with SPARQL property path support.

The highlights of our results are as follows:
\begin{enumerate}
    \item The proposed persistent RPQ evaluation algorithms maintain sub-millisecond edge processing latency  on real-world workloads, and can  process up-to tens of thousands of edges-per-second on a single machine.
    \item The tail (99th percentile) latency of the algorithms increases linearly with the window size $|W|$, confirming the amortized costs  in Table \ref{tab:amortized}.
    \item The cost of expiring old tuples grows linearly with the slide interval $\beta$, which enables constant overhead regardless of $\beta$ when amortized over the slide interval.
    \item Explicit deletions can incur up to 50\% performance degradation on tail latency, however the impact stays relatively steady with the increasing ratio of deletions. 
    \item Although RPQ evaluation under simple path semantics is NP-hard in the worst-case, the results indicate that the majority of the queries formulated on real-world and synthetic streaming graphs can  be evaluated with 2$\times$ to 5$\times$ overhead on the tail latency.
    \item  Our proposed algorithms achieve up to three orders of magnitude better performance when compared to existing RDF systems that emulate stream processing functionalities, substantiating the need for streaming algorithms for persistent RPQ evaluation on streaming graphs.
\end{enumerate}

\subsection{Experimental Setup}
\label{sec:experimental_Setup}

\subsubsection{Implementation}
The prototype system is an in-memory implementation in Java 13 and includes algorithms  in \S \ref{sec:arbitrary} and \S \ref{sec:simple} --- we leave out-of-core processing as future work.
The tree index $\Delta$ is implemented as a concurrent hash-based index where each vertex $v \in G_{W,\tau}$ is mapped to its corresponding spanning tree $T_x$.
Each spanning tree $T_x$ is assisted with an additional hash-based index for efficient node look-ups.
RAPQ (\textbf{\ref{alg:arbitrary}} and \textbf{\ref{alg:insertrapq_expiry}}), RSPQ algorithms (\textbf{\ref{alg:ProcessEdgeRSPQ}}, and \textbf{\ref{alg:insertrapq_expiry}}) employ \textit{intra-query parallelism} by deploying a thread pool to process multiple spanning trees in parallel that are accessed for each incoming edge.
Window management is parallelized similarly.

Experiments are run on a Linux server with 32 physical cores and 256GB memory with the total number of execution threads set to the number of available physical cores.
We measure the time it takes to process each tuple and report the average throughput and the tail latency ($99^{th}$ percentile) after ten minutes of processing on warm caches.
Our prototype implementation is a closed system where each arriving tuple $t_{\tau}$ is processed sequentially.
Thus, the throughput is inversely correlated with the mean latency.

\subsubsection{Workloads and Datasets}
\label{sec:experiments-datasets}
Although there exists streaming RDF benchmarks such as LSBench \cite{lsbench-code}
and Stream WatDiv ~\cite{gao:2018aa}, their workloads do not contain any recursive queries, and they generate streaming graphs with very limited form of recursion.
Therefore,  we formulate persistent RPQs  using the most common recursive queries found in real-world applications,
leveraging recent studies \cite{bonifati2017analytical, bonifati2019navigating} that analyze real-world SPARQL query logs.
We choose the most common 10 recursive queries from \cite{bonifati2019navigating}, which cover more than 99\% of all recursive queries found in Wikidata query logs.
In addition, we choose the most common non-recursive query (with no Kleene stars) for completeness, even though these are easier to evaluate as resulting paths have fixed size.
Table \ref{tab:common_queries} reports the set of real-world RPQs
used in our experiments.
We set $k=3$ for queries with variable number of edge labels as the SO graph only has three distinct labels.
Table \ref{tab:query_bindings} lists the values of edge labels for graphs we used in our experiments.
We run these over the following real and synthetic edge-labeled graphs.


\begin{table}
\caption{The most common RPQs used in real-world workloads (retrieved from Table 4 in \cite{bonifati2019navigating}). }
\small
	\centering
	\begin{tabular}{ | r | c | r | c | }
	\hline
        Name & Query & Name & Query \\
	\hline
	$Q_1$ & $a^*$ & $Q_7$ & $a \circ b \circ c^*$ \\ 
	$Q_2$ & $a \circ b^* $ & $Q_8$ & $a?  \circ b^*$ \\
	$Q_3$ & $a \circ b^*  \circ c^* $ & $Q_{9}$ & $(a_1 + a_2 + \cdots + a_k )^+ $  \\
	$Q_4$ & $(a_1 + a_2 + \cdots + a_k )^*$ & $Q_{10}$ & $(a_1 + a_2 + \cdots + a_k) \circ b^*$  \\
	$Q_5$ & $a \circ b^* \circ c$ & $Q_{11}$ & $a_1 \circ a_2 \circ \cdots \circ a_k$\\
	$Q_6$ & $a^* \circ b^* $& & \\

	\hline
	\end{tabular}
\label{tab:common_queries}
\end{table}

\textbf{Stackoverflow} (SO)  is a temporal graph of user interactions on this website containing 63M interactions (edges) of 2.2M users (vertices), spanning 8 years \cite{paranjape2017motifs}.
Each directed edge $(u,v)$ with timestamp $t$ denotes an interaction between two users: (i) user $u$ answered user $v$'s questions at time $t$, (ii) user $u$ commented on user $v$'s question, or (iii) comment at time $t$.
SO graph is more homogeneous and much more cyclic than other datasets we used in this study as it contains only a single type of vertex and three different edge labels.
7 out of 11 queries in Table \ref{tab:common_queries} have at least 3 labels and cover all edges in the  graph.
Its highly dense and cyclic nature causes a high number of intermediate results and resulting paths; therefore, this graph constitutes the most challenging one for  the proposed algorithms.
We set the window size $|W|$ to 1 month and the slide interval $\beta$ to 1 day unless specified otherwise.

\textbf{LDBC SNB} is synthetic social network graph  that is designed to simulate real-world interactions in social networking applications \cite{sigmod15_erling:2015}.
We extract the update stream of the LDBC workload, which exhibits 8 different types of interactions users can perform.
The streaming graphs generated by LDBC consists of two recursive relations: $knows$ and $replyOf$. Therefore, $Q_4, Q_5, Q_9$ and $Q_{10}$ in Table \ref{tab:common_queries} cannot be meaningfully formulated over the LDBC streaming graphs; we use the others from Table \ref{tab:common_queries}. 
We use a scale factor of 10 with approximately 7.2M users and posts (vertices) and  40M user interactions (edges).
LDBC update stream spans 3.5 months of user activity and we set the window size $|W|$ to 10 days and the slide interval $\beta$ to 1 day unless  specified otherwise.

\textbf{Yago2s}
is a real-world RDF dataset containing 220M triples (edges) with approximately 72M different subjects (vertices) \cite{yago2s-dataset}.
Unlike existing streaming RDF benchmarks, Yago2s includes a rich schema ($\sim$100 different labels) and allows us
to represent the full set of queries listed in Table \ref{tab:common_queries}.
To emulate sliding windows on Yago2s RDF graph, we assign a monotonically non-decreasing timestamp to each RDF triple at a fixed rate.
Thus, each window defined over Yago2s has equal number of edges.
We set the window size $|W|$ such that each window contains approximately 10M edges and the slide interval $\beta$ to  1M edges, unless specified otherwise.


\begin{table}
\caption{ Values of label variables in real-world RPQs (Table \ref{tab:common_queries}) for graphs we tested.}
\small
	\centering
	\begin{tabular}{ | r | c | }
	\hline
        Graph & Predicates \\
	\hline
	SO & \textit{knows}, \textit{replyOf}, \textit{hasCreator}, \textit{likes} \\ 
    LDBC SNB & \textit{a2q}, \textit{c2a}, \textit{c2q} \\ 
    Yago2s & \textit{happenedIn}, \textit{hasCapital}, \textit{participatedIn} \\ 

	\hline
	\end{tabular}
\label{tab:query_bindings}
\end{table}

Additionally, we use {\textbf{gMark}}~{\cite{bagan2016gmark}} graph and query workload generator to systematically analyze the effect of query size $|Q_R|$.
We use a pre-configured schema that mimics the characteristics of LDBC SNB graph to generate a synthetic graph with 100M vertices and 220M edges, and create synthetic query workloads where the query size ranges from 2 to 20 (the size of a query, $|Q_R|$, is the number of labels in the regular expression $R$ and the number of occurrences of $*$ and $+$). 
Each RPQ is formulated by grouping labels into concatenations and alternations of size up to 3 where each group has a 50\% probability of having $*$ and $+$.
As gMark generates the entire LDBC SNB network as a single static graph, we assign a monotonically non-decreasing timestamp to each edge at a fixed rate.

\subsection{Throughput \& Tail Latency}
\label{sec:experiments_analysis}

\begin{figure*}
  \centering
      \subfigure[Yago2s]
      {
        \includegraphics[width=2.1in]{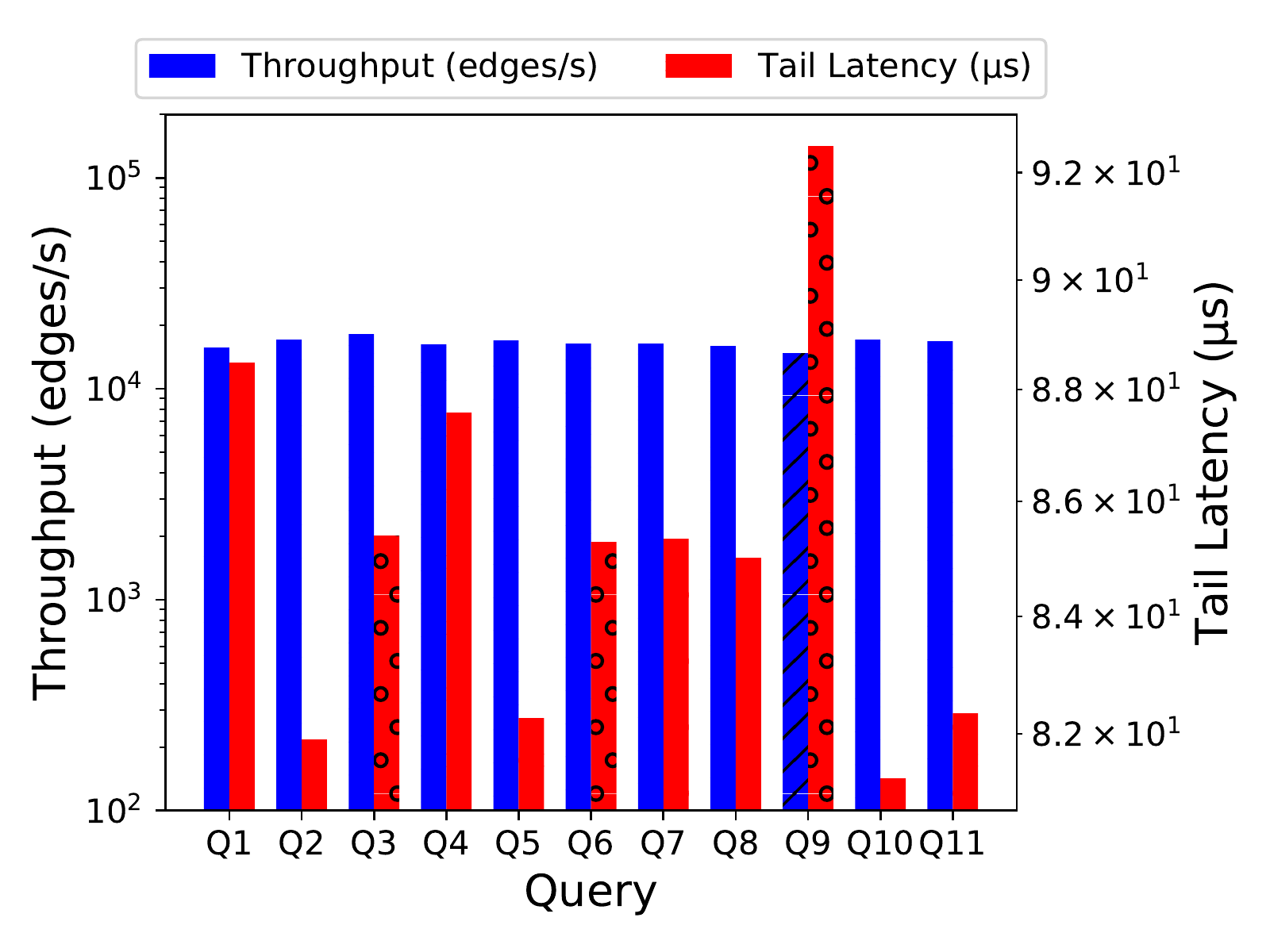}
        \label{fig:yago2s-latency}
    }
    \subfigure[LDBC SF10]
    {
        \includegraphics[width=2.1in]{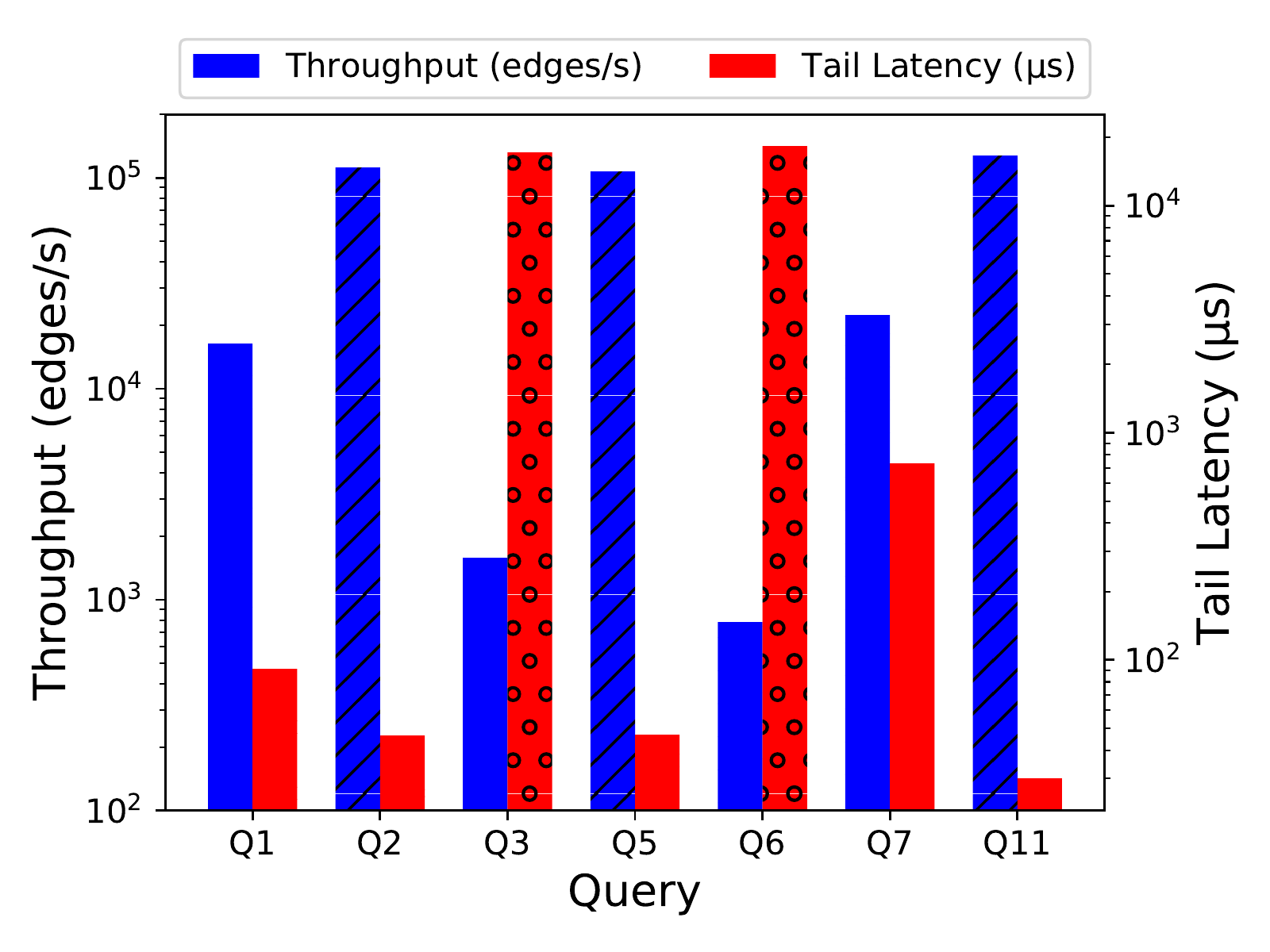}
        \label{fig:ldbcsf10-latency}
    }
    \subfigure[Stackoverflow]
    {
        \includegraphics[width=2.1in]{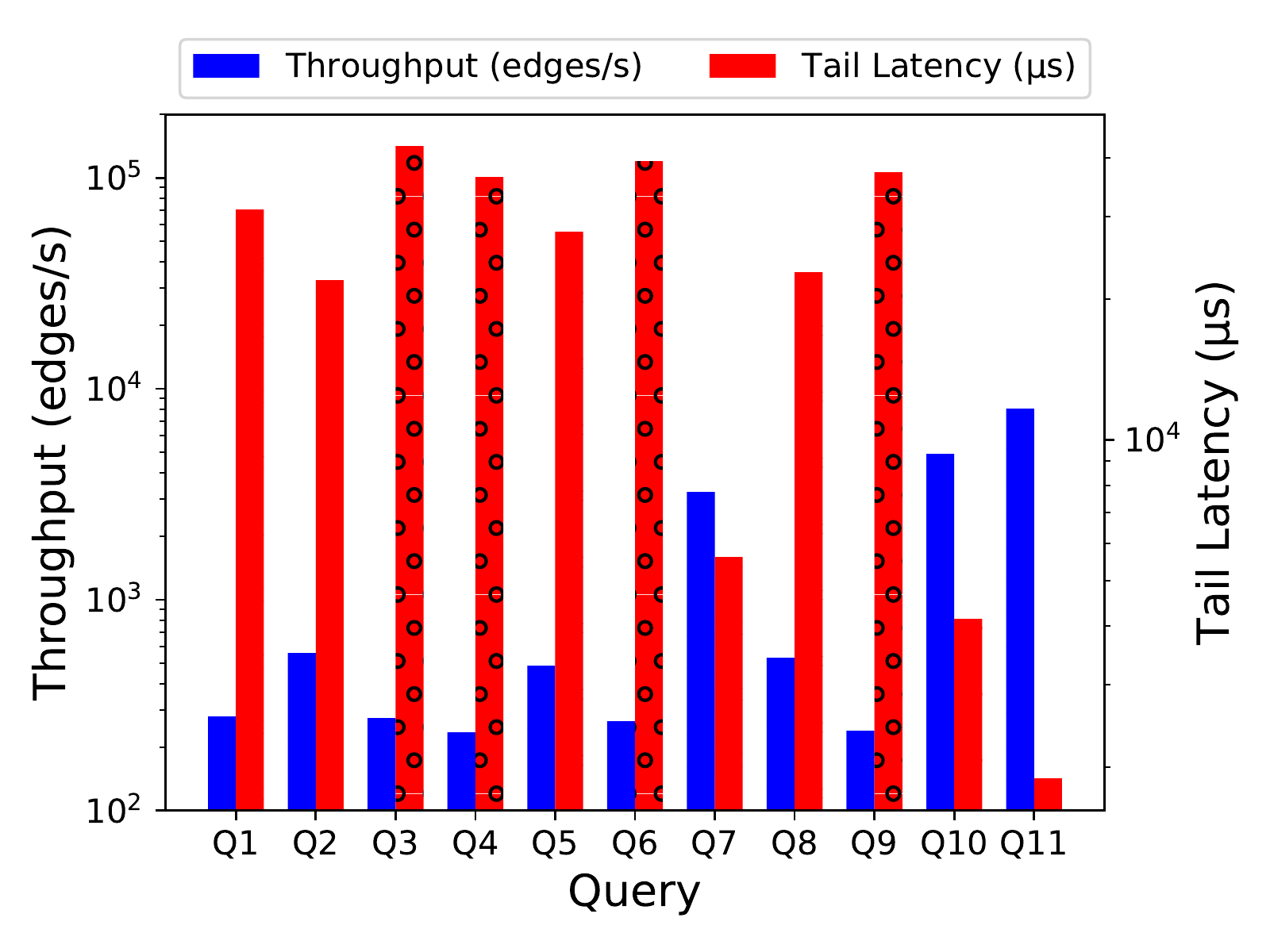}
        \label{fig:stackoverflow-latency}
    }
    \caption{Throughput and tail latency of the Algorithm \textbf{\ref{alg:arbitrary}}. Y axis is given in log-scale.}
    \label{fig:latency-tput}
\end{figure*}

Figure \ref{fig:latency-tput} shows the throughput and tail  latency of Algorithm \textbf{\ref{alg:arbitrary}} for all queries on all datasets.
The algorithm  discards a  tuple  whose label is not in the alphabet $\Sigma_Q$ of $Q_R$ as it cannot be part of any resulting path. 
Hence, we only measure and report  latency of tuples whose labels match a label in the given query.
First, we observe that the performance is generally lower for the SO graph due to its label density and highly cyclic nature.
The tail latency of Algorithm \textbf{\ref{alg:arbitrary}} is below 100ms even for the slowest query $Q_3$ on the SO graph and it is in sub-milliseconds for most queries on Yago2s and LDBC graphs.
Similarly, the throughput of the algorithm varies from hundreds of edges-per-second for the SO graph (Figure \ref{fig:stackoverflow-latency}) to tens of thousands of edges-per-second for LDBC  graph (Figure \ref{fig:ldbcsf10-latency}).

\begin{figure}
  \centering
  \includegraphics[width=2.5in]{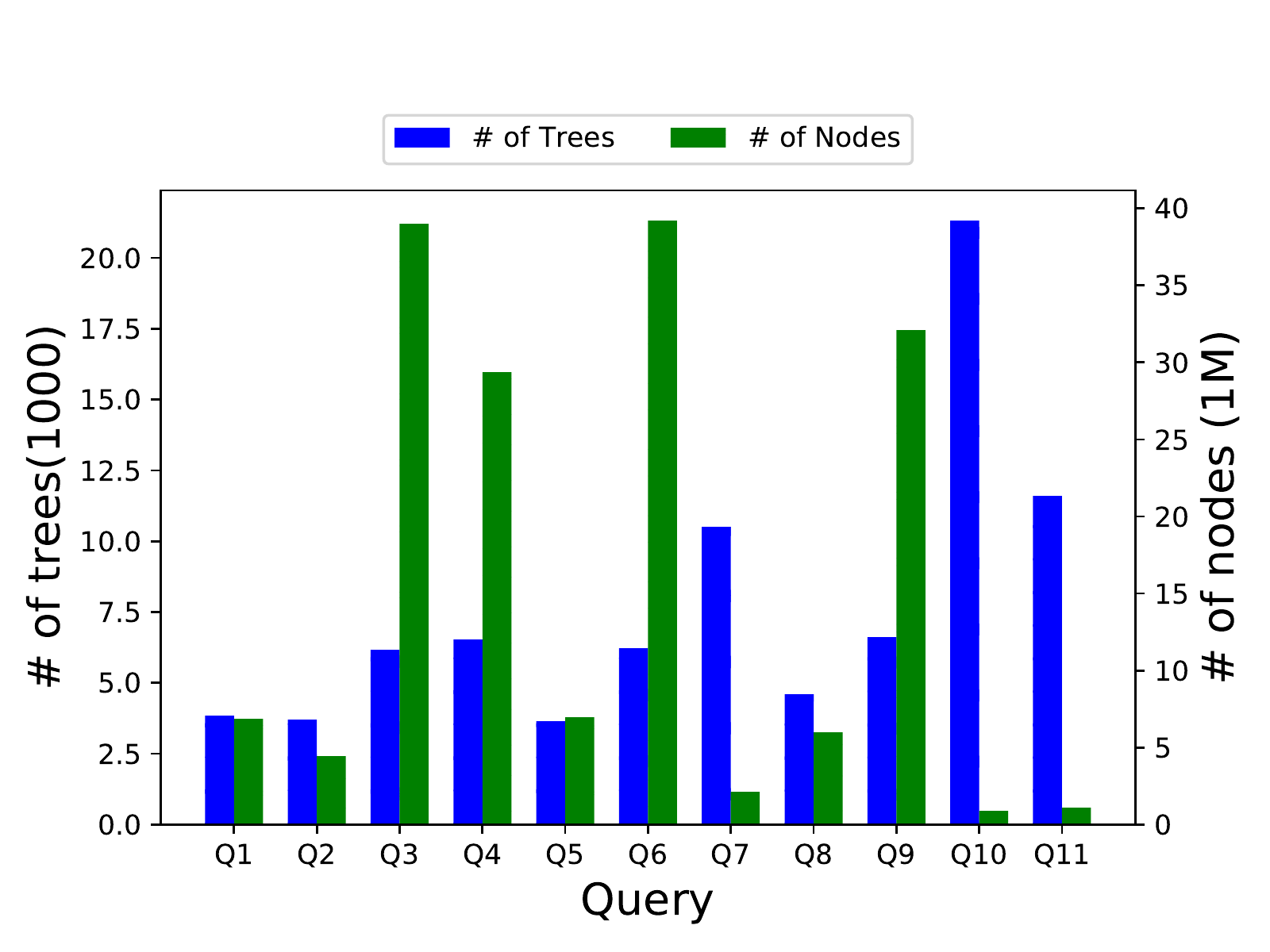}
  \caption{Size of the tree index $\Delta$ on the SO graph.} 
  \label{fig:sx-treesize}
\end{figure}

We plot the total number of trees and nodes in the tree index $\Delta$ of Algorithm \textbf{\ref{alg:arbitrary}} on the SO graph to better understand diverse performance characteristics of different queries.
Remember that nodes and their corresponding paths in a spanning tree $T_x \in \Delta$ represent partial results of a persistent RPQ.
Therefore, the amount of work performed by the algorithm grows with the size of tree index $\Delta$.
As expected, we observe a negative correlation between the throughput of a query (Figure \ref{fig:stackoverflow-latency}) and its tree index size (Figure \ref{fig:sx-treesize}).
It is known that cycles have significant impact on the run time of queries \cite{bonifati2017analytical}, and our analysis confirms this.
In particular, $Q_3$ and $Q_6$ have the largest index sizes and therefore the lowest throughput, which can be explained by the fact that they contain multiple Kleene stars. 
Similarly, $Q_4$ and $Q_9$ have a Kleene star over alternation of symbols, which covers all the edges in the graph as the SO graph has only three types of user interactions.
Therefore, $Q_4$ and $Q_9$ both have large index sizes, which negatively impacts the performance.
In parallel, $Q_{11}$ has the highest throughput on all datasets as it is the only fixed size, non-recursive query employed in our experiments.

\subsection{Scalability \& Sensitivity Analysis}
\label{sec:experiments_scalability}

\begin{figure*}
\centering
    \subfigure[Tail Latency]
    {
        \includegraphics[width=3.2in]{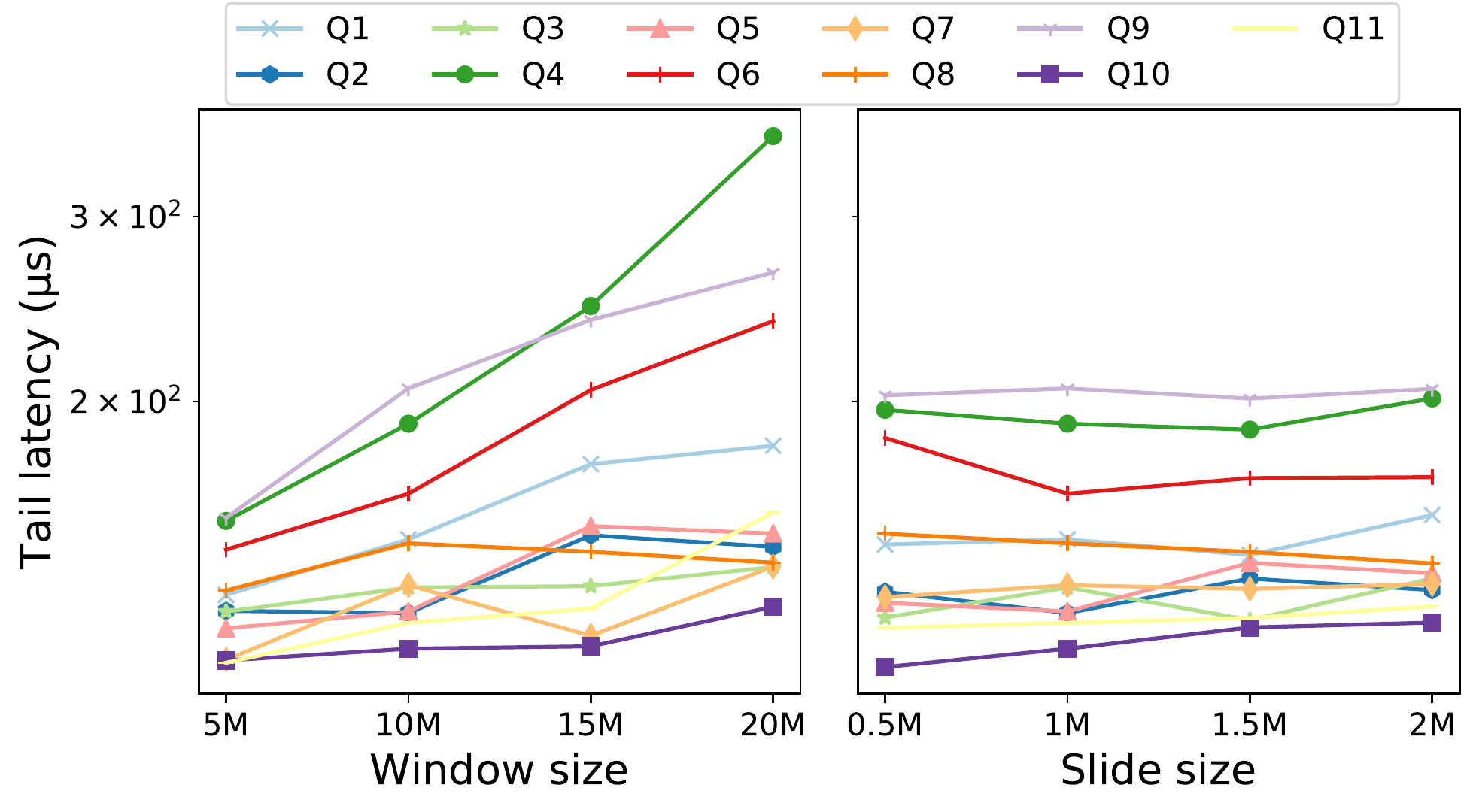}
        \label{fig:yago2s-window-slide}
    }
    \subfigure[Window Management Time]
    {
        \includegraphics[width=3.2in]{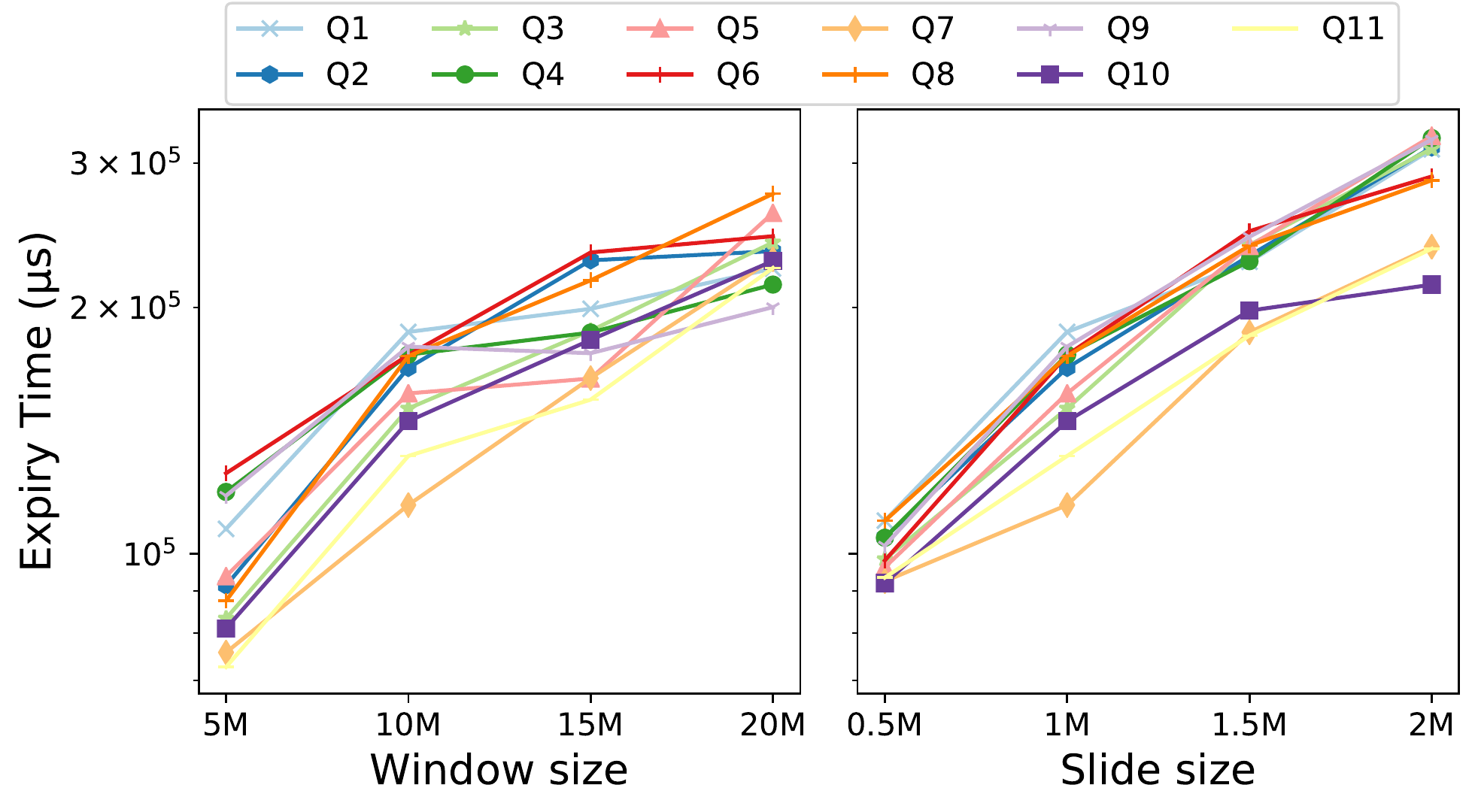}
        \label{fig:yago2s-expiry}
    }
    \caption{The tail latency (a) and the average window maintenance cost (b) with various $|W|$ and  $\beta$.} 
\end{figure*}


In this section, we first assess the impact of the window size $|W|$ and the slide interval $\beta$ on algorithm performance; then, we turn our attention to performance implications of the use of DFAs and the query size $|Q_R|$.

We use the Yago2s dataset for this experiment as windows with a fixed number of edges we created over  Yago2s enable us to precisely assess the impact of window size.
Figure \ref{fig:yago2s-window-slide} presents the tail latency of our algorithm where the window size changes from 5M edges to 20M edges with 5M intervals.
As expected, the tail latency for all queries we tested increases with increasing $|W|$, which conforms with the amortized cost analysis of Algorithm \textbf{\ref{alg:arbitrary}} in \S \ref{sec:arbitrary-insertonly}.
Similarly, we observe that the time spent on Algorithm \textbf{\ref{alg:insertrapq_expiry}} increases with increasing window size $|W|$ (Figure \ref{fig:yago2s-expiry}), in line with the complexity analysis given in \S \ref{sec:arbitrary-insertonly}.
We replicate the same experiment using LDBC and Stream WatDiv datasets by varying the scale factor which in turn increases the number of edges in each window.
Our results show a degradation on the performance with increasing scale factor on Stream WatDiv, confirming our findings on Yago2s.
However, we do not observe a similar trend on LDBC graphs, which is due to the linear scaling of the total number of edges and vertices with the scale factor. 
Increasing the scale factor reduces the density of the graph, which may cause the proposed algorithms to perform even better in some instances due to a smaller tree index size.
Furthermore, only a subset of queries can be formulated on these datasets as described previously.
Therefore, we only report our findings on Yago2s graph.


Next, we assess the impact of the slide interval $\beta$ on the performance of our algorithms.
Figure \ref{fig:yago2s-window-slide} plots the tail latency of Algorithm \textbf{\ref{alg:arbitrary}} against  $\beta$ and shows that the slide interval does not impact the performance.
Recall that Algorithm \textbf{\ref{alg:insertrapq_expiry}} is invoked periodically to remove expired tuples from the tree index $\Delta$.
It first identifies the set of expired nodes in a given spanning tree $T_x \in \Delta$, and searches their incoming edges to find a valid edge from  a valid node in $T_x$.
Therefore, Algorithm \textbf{\ref{alg:insertrapq_expiry}} might traverse the entire snapshot graph $G_{W,\tau}$ in the worst-case, regardless of the slide interval $\beta$.
However, Figure \ref{fig:yago2s-expiry} shows that the time spent on expiry of old tuples grows with increasing $\beta$, which causes its overhead to stay constant over time regardless of the slide interval $\beta$.
Therefore, this algorithm  is robust to the slide interval $\beta$. 
It also suggests that the complexity analysis  of Algorithm \textbf{\ref{alg:insertrapq_expiry}} given in \S \ref{sec:arbitrary-insertonly} is not tight.

\begin{figure}
    \centering
    \includegraphics[width=2.5in]{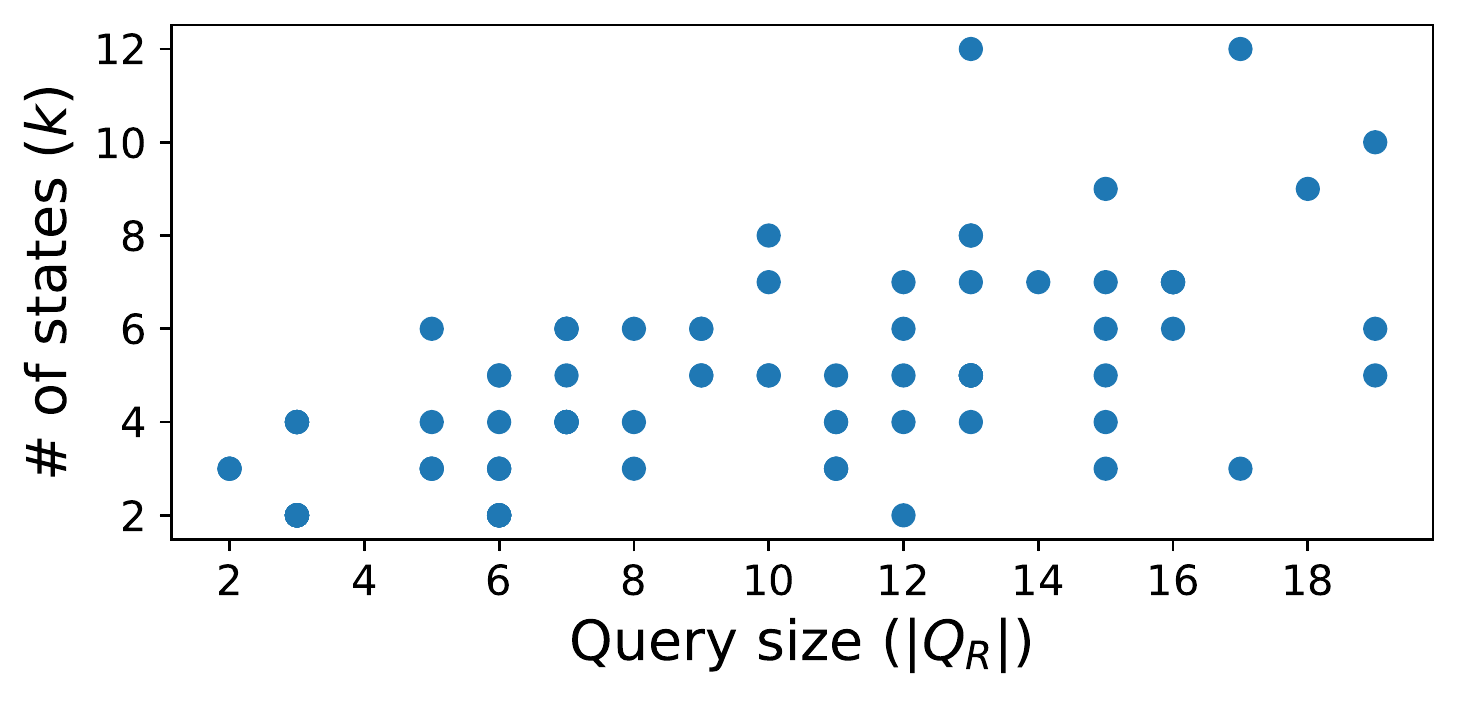}
    \caption{The number of states $k$ in corresponding DFAs of queries in the synthetic RPQ workload.}
    \label{fig:length-states}
\end{figure}

Finally, we analyze the effect of the query size $|Q_R|$ and the automata size $k$ on the performance of our algorithms using a set of 100 synthetic RPQs that are generated using gMark. 
Combined complexities of the algorithms presented in {\S \ref{sec:arbitrary}} and {\S \ref{sec:simple}} are polynomial in the number of states $k$, which might be exponential in the query size $|Q_R|$. 
Figure {\ref{fig:length-states}} shows the total number of states in minimized DFAs for 100 RPQs we created using gMark; in practice, we found out that the size of the DFA does not grow exponentially with increasing query size for the considered RPQs despite the theoretical upper bound.
Green et al. \cite{green2003processing} has also indicated that exponential DFA growth is of little concern for most practical applications in the context of XML stream processing.

Next, we focus on the impact of the automata size $k$ on  performance.
Figure {\ref{fig:tput-states}} plots the throughput against the number of states $k$ in the minimal automata for synthetic RPQs generated by gMark.
We do not observe a significant impact of $k$ on  performance; yet, performance differences for queries with the same number of states in their corresponding DFA can be up to $6\times$. 
Such performance difference for RPQ evaluation has already been observed on static graphs and has been attributed to query label selectivities and the size of intermediate results {\cite{yakovets2016query}}. 
To further verify this hypothesis in the streaming model, we plot the throughput against the tree index $\Delta$ size for queries with $k=5$ in Figure {\ref{fig:tput-delta}}.
Confirming our results in {\S} {\ref{sec:experiments_analysis}}, we observe a negative correlation between the throughput of a query and its tree index size.

\begin{figure}
    \centering
    \includegraphics[width=2.5in]{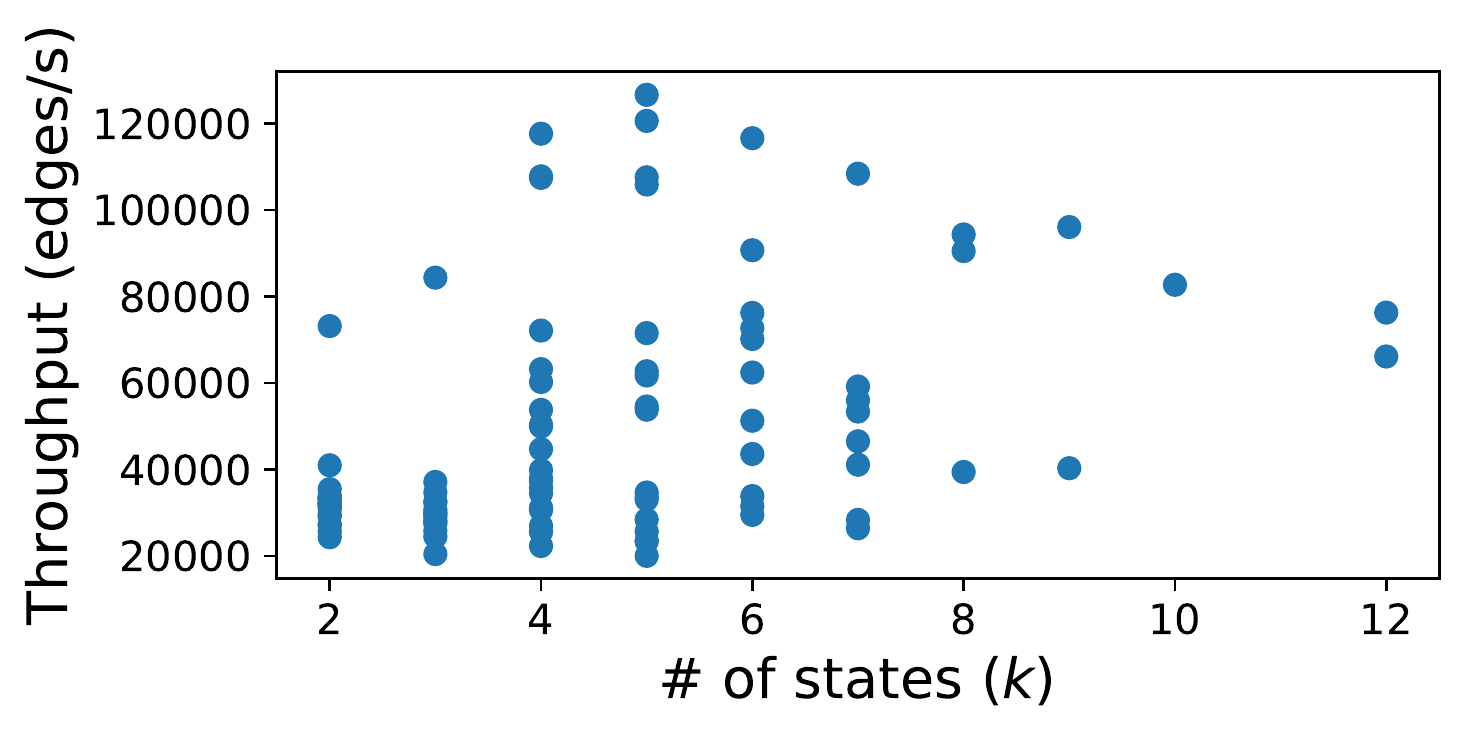}
    \caption{Throughput of the Algorithm \ref{alg:arbitrary} for the synthetic RPQ workload.}
    \label{fig:tput-states}
\end{figure}

\begin{figure}
    \centering
    \includegraphics[width=2.5in]{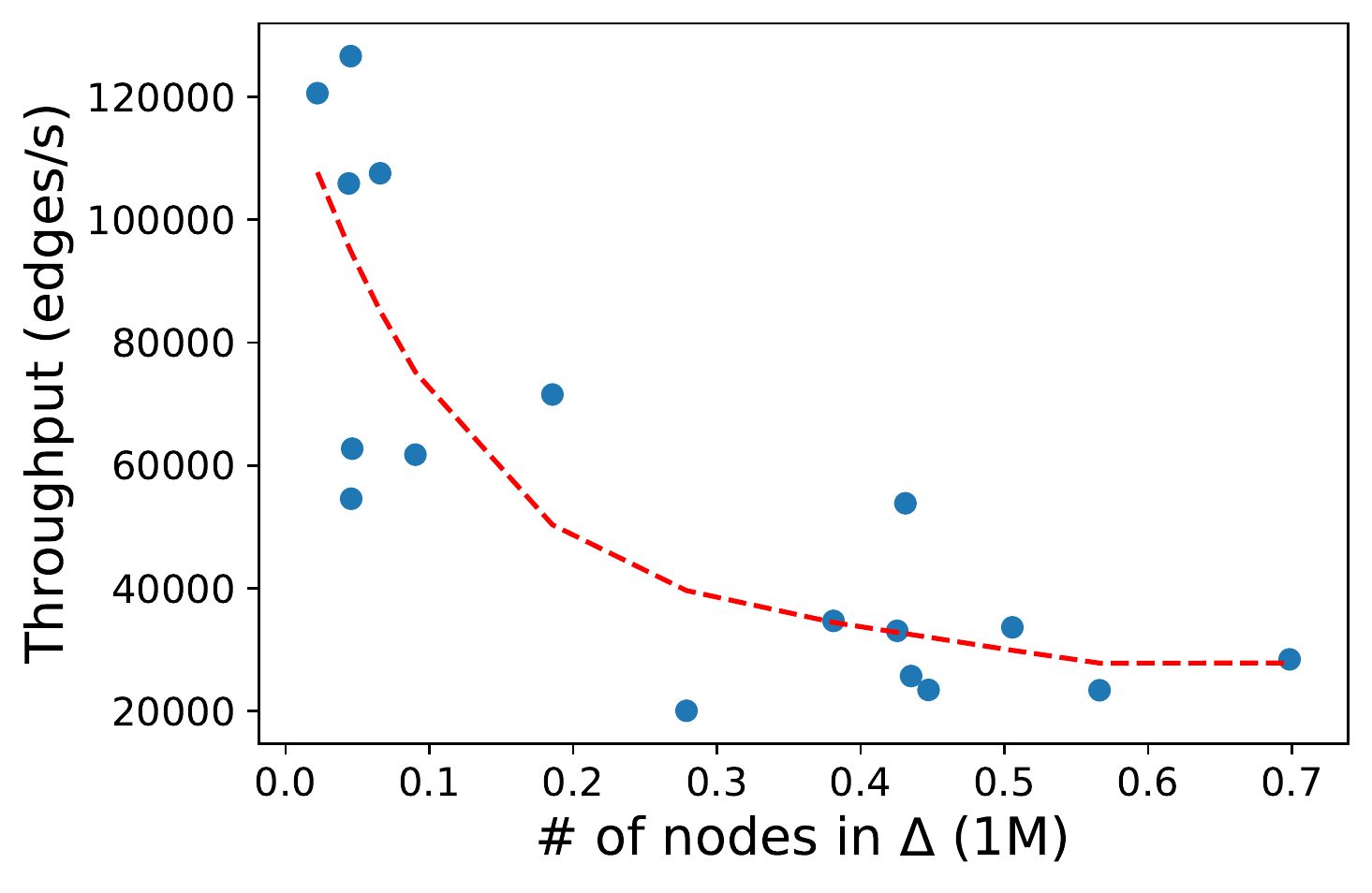}
    \caption{Throughput and tree index $\Delta$ size for synthetic RPQs with $k=5$ }
    \label{fig:tput-delta}
\end{figure}

\vspace{-4mm}
\subsection{Explicit Edge Deletions}
\label{sec:experiments_deletion}

\begin{figure}
    \centering
    \includegraphics[width=2.5in]{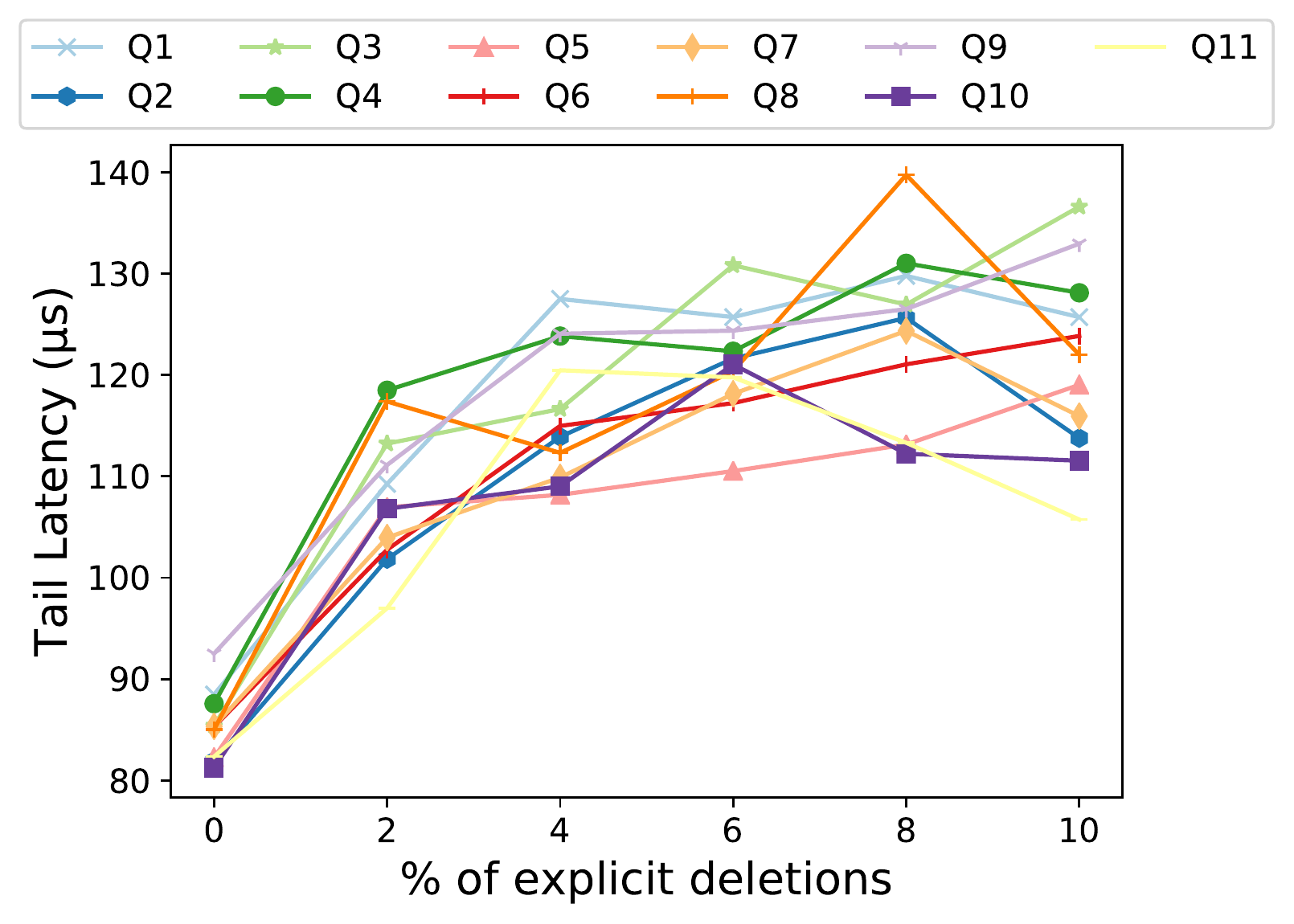}
    \caption{Impact of the ratio of explicit deletions  on tail latency for all queries on Yago2s RDF graph.}
    \label{fig:delete}
\end{figure}

Although most real-life streaming graphs are append-only, some applications require explicit edge deletions, which can be processed in our framework (\S \ref{sec:arbitrary-delete}). 
We generate explicit deletions by reinserting a previously consumed edge as a negative tuple and varying the ratio of negative tuples
in the stream.
Figure \ref{fig:delete} plots tail latency of all queries on Yago2s  varying deletion ratio from 2\% to 10\%.
In line with our findings in the previous section,  explicit deletions incur performance degradation due to the overhead of the expiry procedure (Figure \ref{fig:yago2s-expiry}).
However, this overhead quickly flattens and does not increase with the deletion ratio.
This is explained by the fact that the sizes of the snapshot graph $G_{W,\tau}$, and the tree index $\Delta$ decrease with increasing deletion ratio.

\subsection{RPQ under Simple Path Semantics}
\label{sec:experiments_simplepath}

We showed (\S \ref{sec:simple}) that the amortized time complexity of Algorithm \textbf{\ref{alg:ProcessEdgeRSPQ}} under simple path semantics is the same as its RAPQ counterpart in the absence of conflicts.

\begin{table}
    \small
    \caption{Queries that can be evaluated under simple path semantics \& the relative slowdown.}
    \begin{tabular}{|c|c|c|}
        \hline
        Graph & \makecell[c]{Succesfull Queries} & \makecell[c]{Latency Overhead} \\
        \hline
        Yago2s & All & $1.8\times - 2.1\times$ \\
        Stackoverflow & $Q_1, Q_4, Q_7, Q_{10}, Q_{11}$ & $1.4\times - 5.4\times$ \\
        LDBC SF10 & $Q_1, Q_2, Q_5, Q_7, Q_{11}$ & $1.8\times - 3\times$\\
        \hline
    \end{tabular}

    \label{tab:rspq_overhead}
\end{table}

In this section, we empirically analyze the feasibility and the performance of this algorithm.
Table \ref{tab:rspq_overhead} lists the queries that can be successfully evaluated under simple path semantics on each graph. 
$Q_1$, $Q_4$ and $Q_{11}$ are restricted regular expressions, a condition that implies conflict-freedom in any arbitrary graph.
Therefore, these queries are successfully evaluated on all graphs we tested (except $Q_4$ that cannot be defined over LDBC graph as discussed in \S \ref{sec:experiments-datasets}).
In particular, we observe that all queries are free of conflicts on Yago2s, and they can successfully be evaluated.

Table \ref{tab:rspq_overhead} also reports the overhead of enforcing simple path semantics on the tail  latency.
This overhead is simply due to conflict detection and the maintenance of markings for each spanning tree in the tree index $\Delta$.
Overall, these results suggest the feasibility of enforcing simple path semantics for majority of real-world queries, considering that most queries are conflict-free on heterogeneous, sparse graphs such as RDF graphs and social networks. 
Conversely, we argue that arbitrary path semantics may be the only practical alternative for applications with homogeneous, highly cyclic graphs such as communication networks like Stackoverflow.


\subsection{Comparison with Other Systems}
\label{sec:experiments_comparison}

This is the first work that investigates the execution of persistent RPQs over streaming graphs; therefore, there are no systems with which a direct comparison can be performed. 
However, there are a number of streaming RDF systems that can potentially be considered. 
These were reviewed in \S \ref{sec:related}; unfortunately, as noted in that section, these systems only support SPARQL v1.0 and therefore cannot handle path expressions or recursive queries.

With the introduction of property paths in SPARQL v1.1, the support for path queries have been added to a few RDF systems such as Virtuoso~{\cite{erling2009}} and RDF-3X~{\cite{gubichev2013sparqling,gubichev2015query}}. 
However, these RDF systems are designed for static RDF datasets, and they do not support persistent query evaluation.
We emulate persistent queries over Virtuoso to highlight the benefit of using incremental algorithms for persistent query evaluation on streaming graphs.

\begin{figure}
    \centering
    \includegraphics[width=2.5in]{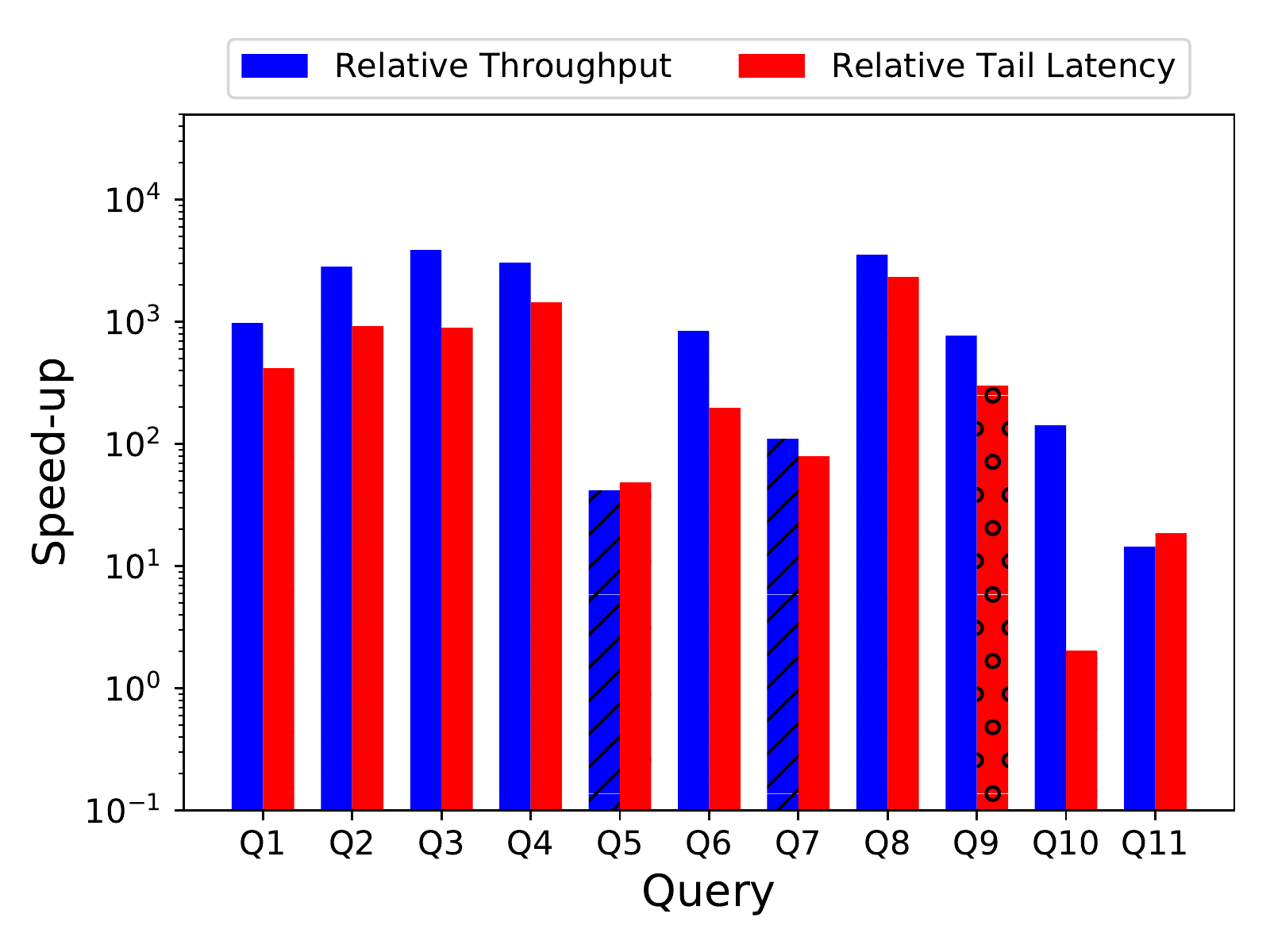}
    \caption{Relative speed-up of Algorithm \ref{alg:arbitrary} over Virtuoso for all queries on Yago2s RDF graph. Y axis is given in log-scale.}
    \label{fig:comparison-yago2s}
\end{figure}

We develop a middle layer on top of Virtuoso that emulates persistent query evaluation over sliding windows, similar to Algorithm \textbf{\ref{alg:arbitrary}}.
This layer inserts each incoming tuple into Virtuoso and evaluates the query on the RDF graph that is constructed from the content of the window $W$ at any given time $t$.
For fairness, we configure Virtuoso to work entirely in memory and disable transaction logging to eliminate the overhead of transaction processing. 
We use Yago2s RDF graph with default $|W|$ and $\beta$ for this experiment.
We need to modify $Q_1, Q_4, Q_6, Q_8, Q_9$ and $Q_{10}$ by prepending a single predicate $a$ to each query due to Virtuoso's limitation forbidding vertex variables on both ends of property paths at the same time.
Figure {\ref{fig:comparison-yago2s}} plots the average speed-up of \textbf{\ref{alg:arbitrary}} with respect to this simulation for both throughput and tail-latency.
\textbf{\ref{alg:arbitrary}} consistently outperforms Virtuoso across all queries and provide up to 3 orders of magnitude better throughput and tail latency.
This is because Virtuoso re-evaluates the RPQ on the entire window and cannot utilize the results of previous computations.
Conversely, \textbf{\ref{alg:arbitrary}} indexes traversals in $\Delta$ and only explores the part of the snapshot graph $G_{W,\tau}$ that were not previously explored.
In summary, these results suggest that incremental evaluation as in the proposed algorithms have significant performance advantages in executing RPQs over streaming graphs.


\section{Related Work}
\label{sec:related}

\textbf{Stream Processing Systems:}
Early research on stream processing primarily adopt the relational model and its query operators in the streaming settings
(STREAM ~\cite{abw02}, Aurora ~\cite{abadi2003}, Borealis ~\cite{cidr05_abadiabcchlmrrtxz05}).
Whereas, modern Data Stream Processing Systems (DSPS) such as Storm ~\cite{toshniwal2014storm}, Heron \cite{sigmod15_kulkarni:2015}, Flink \cite{carbonekemht15}
are mostly scale-out solutions that do not necessarily offer a full set of DBMS functionality.
Existing literature  (as surveyed by Hirzel et al. \cite{HirzelBBVSV18}) heavily focus on general-purpose systems and do not consider core graph querying functionality such as \textit{subgraph pattern matching} and \textit{path navigation}.

There has been a significant amount of work on various aspects of RDF stream processing\footnote{\url{https://www.w3.org/community/rsp/wiki/Main_Page}}. 
Calbimonte~\cite{calbimonte2017linked} designs a communication interface for streaming RDF systems based on the Linked Data Notification protocol. 
TripleWave~{\cite{mauri2016triplewave}}  focuses on the problem of RDF stream deployment and introduces a framework for publishing RDF streams on the web. 
EP-SPARQL~{\cite{anicic2011ep}} extends SPARQLv1.0  for reasoning and a complex event pattern matching on RDF streams.
Similarly, SparkWave~{\cite{komazec2012sparkwave}} is designed for streaming reasoning with schema-enhanced graph pattern matching and relies on the existence of RDF schemas to compute entailments.
None of these are processing engines, so they do not provide query processing capabilities. 
Most similar to ours are streaming RDF systems with various SPARQL extensions for persistent query evaluation over RDF streams such as C-SPARQL \cite{barbieri2009c}, CQELS \cite{le2011native}, SPARQL$_{stream}$~\cite{calbimonte2010enabling} and W3C proposal RSP-QL \cite{dell2015towards}. 
However, these systems are designed for SPARQLv1.0, and they do not have the notion of \textit{property path}s from SPARQLv1.1. Thus one cannot formulate path expressions such as RPQs that cover more than 99\% of all recursive queries abundantly found in massive Wikidata query logs {\cite{bonifati2019navigating}}.
The lack of property path support of these systems is previously reported by an independent RDF streaming benchmark, SR-Bench~{\cite{zhang:2012aa}} (see Table 3 in {\cite{zhang:2012aa}}).
Furthermore, query processing engines of these systems do not employ incremental operators, except 
Sparkwave~\cite{komazec2012sparkwave} that focuses on stream reasoning. 
On the contrary, our proposed algorithms incrementally maintain results for a persistent query $Q_R$ as the graph edges arrive.
Our contributions are orthogonal to existing work on streaming RDF systems, although the algorithms proposed in this paper can be integrated into these systems as they incorporate SPARQLv1.1 (i.e., property paths) to provide native RPQ support.

\noindent
\textbf{Streaming \& Dynamic Graph Theory:}
Earlier work on streaming graph algorithms is motivated by the limitations of main memory,
and existing literature has widely adopted the \textit{semi-streaming} model for graphs where the set of vertices can be stored in memory but not the set of edges \cite{muthukrishnan2005data}, due to infeasibility of graph problems in sublinear space.
There exist a plethora of approximation algorithms in this model, and we refer interested readers to \cite{mcgregor2014graph} for a survey.

Graph problems are widely studied in the dynamic graph model where algorithms may use the necessary memory to store the entire graph and compute how the output changes as the graph is updated.
Examples include connectivity \cite{kapron2013dynamic}, shortest path \cite{bernstein2016maintaining}, transitive closure \cite{lkacki2011improved}.
Most related to ours is dynamic reachability, which can be used to solve RPQ under arbitrary path semantics given the entire product graph (Definition \ref{def:product_graph}).
The state-of-the-art dynamic reachability algorithm has $\mathcal{O}(m+n)$ amortized update time \cite{roditty2016fully}.
Our proposed algorithms have a lower amortized cost,  $\mathcal{O}(n)$, for insertions at the expense of  $\mathcal{O}(n^2)$ amortized time for deletions -- a trade-off justified by the insert-heavy nature of real-world streaming graphs. 
Fan et al. \cite{fan2017incremental} characterize the complexity of various graph problems, including RPQ evaluation, in the dynamic model and show that most graph problems are unbounded under edge updates, i.e., the cost of computing changes to query answers cannot be expressed as a polynomial of the size of the changes in the input and output.
They prove that RPQ is bounded relative to its batch counterpart; the batch algorithm can be efficiently incrementalized by minimizing unnecessary computation.

\noindent
\textbf{Regular Path Queries:}
The research on RPQs focuses on various problems such as containment \cite{calvanese2000query}, enumeration \cite{martens2017enumeration}, learnability \cite{bonifati:hal-01068055}. 
Most related to ours is the RPQ evaluation problem.
The seminal work of Mendelzon and Wood \cite{mendelzon1995finding} shows that RPQ evaluation under simple path semantics is NP-hard for arbitrary graphs and queries.
They identify the conditions for graphs and regular languages where the 
introduce a maximal class of regular languages, $C_{tract}$, for which the problem of RPQ evaluation under simple path semantics is tractable. 

RPQ evaluation strategies follow two main approaches: automata-based and relational algebra-based.
$\mathbf{G}$ \cite{cruz1987graphical}, one of the earliest graph query languages,
builds a finite automaton from a given RPQ to guide the traversal on the graph.
Kochut et al.  \cite{kochut2007sparqler} 
study  RPQ evaluation in the context of SPARQL and propose an algorithm that uses two automatons, one for the original expression and one for the reversed expression, to guide a bidirectional BFS on the graph.
Addressing the memory overhead of BFS traversals, 
Koschmieder et al. \cite{koschmieder2012regular} decompose a query into smaller fragments based on rare labels and perform a series of bidirectional searches to answer individual subqueries.
A recent work by Wadhwa et al. \cite{Wadhwa:2019:EAR:3299869.3319882} uses random walk-based sampling for approximate RPQ evaluation.
The other alternative for RPQ evaluation is $\alpha$-RA that extends the standard relational algebra with the $\alpha$ operator for transitive closure computation \cite{agrawal1988alpha}. 
$\alpha$-RA-based RPQ evaluation strategies are used in various SPARQL engines \cite{erling2009}.
Histogram-based path indexes on top of a relational engine can speed-up processing RPQs with bounded length \cite{DBLP:conf/edbt/FletcherPP16}.
$\alpha$-RA-based RPQ evaluation is not suitable for persistent RPQ evaluation on streaming graphs as it relies on blocking join and $\alpha$ operators.
Hence, we adapt the automata-based RPQ evaluation in this paper and introduce non-blocking, incremental algorithms for persistent RPQ evaluation.
Besides, Yakovets et al. \cite{yakovets2016query} show that these two approaches are incomparable and they can be combined to explore a larger plan space for SPARQL evaluation.
Various formalisms such as pebble automata, register automata, monadic second-order logic with data comparisons extend RPQs with data values for the property graph model \cite{libkin2016querying,libkin2012regular}.
Although RPQs and corresponding evaluation methods are widely used in graph querying \cite{angles2017foundations, angles2018g, erling2009}, all of these works
focus on static graphs; ours is, to the best of our knowledge, the first work to consider persistent RPQ evaluation on streaming graphs.


\section{Conclusion and Future Work}
\label{sec:conclusion}
In this paper, for the first time, we study the problem of efficient persistent RPQ evaluation on sliding windows over streaming graphs.
The proposed algorithms process explicit edge deletions under both arbitrary and simple path semantics in a uniform manner.
In particular, the algorithm for simple path semantics has the same complexity as the algorithm for arbitrary path semantics in the absence of conflicts, and it admits efficient solutions under the same condition as the batch algorithm.
Experimental analyses using a variety of real-world RPQs and streaming graphs show that proposed algorithms can support up to tens of thousands of edges-per-second while maintaining sub-second tail latency.
Future research directions we consider in this project are: (i) to extend our algorithms with attribute-based predicates to fully support the popular property graph data model, and (ii) to investigate multi-query optimization techniques to share computation across multiple persistent RPQs.

\begin{acks}
This research was partially supported by grants from Natural Sciences and Engineering Research Council (NSERC) of Canada and  Waterloo-Huawei Joint Innovation Lab.
This research started during Angela Bonifati's sabbatical leave (supported by INRIA) at the University of Waterloo in 2019. 
\end{acks}


\bibliographystyle{ACM-Reference-Format}
\bibliography{publishers,publications,references_rpq}

\clearpage

\end{document}